\documentclass[twocolumn,10pt]{asme2ej}

\usepackage{amsthm,amsmath,amssymb,amsfonts}
\usepackage{array}
\usepackage{graphicx}
\usepackage{mathtools}
\usepackage{hyperref}   
\usepackage{subcaption}
\hypersetup{
	colorlinks=true,
	linkcolor=blue,
	citecolor=blue,
	urlcolor=blue,
}
\usepackage[square,numbers]{natbib}
\newtheorem*{remark}{Remark}
\newtheorem{proposition}{Proposition}

\newtheorem{defn}{Definition}
\newtheorem{lemma}{Lemma}
\newcommand{\revres}{\textcolor{black}} 

\title{Design and validation of a state-dependent Riccati equation filter for state of charge estimation in a latent thermal storage device}

\author{Michael Shanks\\
    School of Mechanical Engineering\\
    Purdue University \\
    West Lafayette, Indiana, 47907\\
    shanks5@purdue.edu
}

\author{Uduak Inyang-Udoh\\
    School of Mechanical Engineering\\
    Purdue University \\
    West Lafayette, Indiana, 47907\\
    uinyangu@purdue.edu
}

\author{Neera Jain\\
    School of Mechanical Engineering\\
    Purdue University \\
    West Lafayette, Indiana, 47907\\
    neerajain@purdue.edu
}

\begin{document}

\maketitle    

\begin{abstract}
{\it\indent Latent thermal energy storage (TES) devices could enable advances in many thermal management applications, including peak load shifting for reducing energy demand and cost of HVAC or providing supplemental heat rejection in transient thermal management systems. However, real-time feedback control of such devices is currently limited by the absence of suitable state of charge estimation techniques, given the nonlinearities associated with phase change dynamics. In this paper we design and experimentally validate a state-dependent Riccati equation (SDRE) filter for state of charge estimation in a phase change material (PCM)-based TES device integrated into a single-phase thermal-fluid loop. The advantage of the SDRE filter is that it does not require linearization of the nonlinear finite-volume model; instead, it uses a linear parameter-varying system model which can be quickly derived using graph-based methods. We leverage graph-based methods to prove that the system model is uniformly detectable, guaranteeing that the state estimates are bounded.  Using measurements from five thermocouples embedded in the PCM of the TES and two thermocouples measuring the fluid temperature at the inlet and outlet of the device, the state estimator uses a reduced-order finite-volume model to determine the temperature distribution inside the PCM and in turn, the state of charge of the device.  We demonstrate the state estimator in simulation and on experimental data collected from a thermal management system testbed to show that the state estimation error converges near zero and remains bounded.
}
\end{abstract}

\begin{nomenclature}
\entry{$a$}{Area $\left[\textup{m}^2\right]$}
\entry{$A$}{State-space coefficient matrix}
\entry{$\mathcal{A}$}{Set of adjacent control volumes}
\entry{$B$}{State-space input matrix}
\entry{$c_p$}{Specific heat $\left[\frac{\textup{J}}{\textup{kg-K}}\right]$}
\entry{$C$}{State-space output matrix}
\entry{$\mathcal{C}$}{Set of PCM control volumes}
\entry{$d$}{Distance $\left[\textup{m}\right]$}
\entry{$e$}{Estimation error}
\entry{$e_{rms}$}{Root-mean-square error}
\entry{$f_m$}{Melt fraction of a PCM control volume}
\entry{$\mathcal{F}$}{Set of fluid control volumes}
\entry{$h$}{Specific enthalpy $\left[\frac{\textup{J}}{\textup{kg}}\right]$}
\entry{$h_{fus}$}{Specific enthalpy of fusion $\left[\frac{\textup{J}}{\textup{kg}}\right]$}
\entry{$H$}{Enthalpy $\left[\textup{J}\right]$}
\entry{$i,j,l$}{Control volume indices}
\entry{$I_n$}{$n\times n$ identity matrix}
\entry{$k$}{Discrete time step index}
\entry{$K$}{Kalman gain}
\entry{$L$}{Graph Laplacian matrix}
\entry{$m$}{Mass $\left[\textup{kg}\right]$}
\entry{$\dot{m}$}{Mass flow rate $\left[\frac{\textup{kg}}{\textup{s}}\right]$}
\entry{$M$}{Capacitance matrix}
\entry{$n$}{Number of control volumes}
\entry{$N_s$}{Number of time steps}
\entry{$\mathcal{N}$}{Gaussian random variable}
\entry{$p$}{Number of outputs}
\entry{$P$}{State estimate error covariance}
\entry{$\mathcal{P}$}{Set of metal plate control volumes}
\entry{$\dot{Q}$}{Heat transfer rate $\left[\textup{W}\right]$}
\entry{$R$}{Thermal resistance $\left[\frac{\textup{K}}{\textup{W}}\right]$}
\entry{$t$}{Time $\left[\textup{s}\right]$}
\entry{$T$}{Temperature $\left[\textup{K}\right]$}
\entry{$T_{pc}$}{Phase-change temperature $\left[\textup{K}\right]$}
\entry{$u$}{Input vector}
\entry{$U$}{Convective heat transfer coefficient $\left[\frac{\textup{W}}{\textup{m}^2\textup{-K}}\right]$}
\entry{$v$}{Measurement noise vector}
\entry{$V$}{Measurement noise covariance}
\entry{$w$}{Process noise vector}
\entry{$W$}{Process noise covariance}
\entry{$x$}{State vector}
\entry{$\hat{x}$}{Estimated state vector}
\entry{$x_{soc}$}{State of charge}
\entry{$y$}{Output or measurement vector}
\entry{$\alpha$}{Specific heat function width parameter}
\entry{$\Gamma$}{Discrete-time input matrix}
\entry{$\kappa$}{Thermal conductivity $\left[\frac{\textup{W}}{\textup{m-K}}\right]$}
\entry{$\Phi$}{Discrete-time state transition matrix}
\end{nomenclature}


\section{Introduction} \label{sec:intro}

There is growing interest in the integration of thermal energy storage (TES) devices with a variety of thermal-fluid systems for improving performance. However, new control strategies are needed to take full advantage of the potential benefits of TES devices in fast-timescale transient thermal management systems (TMSs) \cite{shanks_control_2022, shafiei_model_2015, pangborn_hierarchical_2020, barz_state_2018, zsembinszki_evaluation_2020}. Control strategies for many forms of TMSs with integrated latent thermal storage, or \emph{hybrid} thermal management systems, have been studied in both experimental and simulation environments.  In \cite{shanks_control_2022}, Shanks et al.\ demonstrate a logic-based controller for state of charge management in a simulated single-phase hybrid TMS for aircraft electronics cooling.  In \cite{shafiei_model_2015}, Shafiei and Alleyne develop a model predictive controller in simulation for a two-phase hybrid TMS in a transport refrigeration system. Similarly, Pangborn et al.\ develop a model predictive controller for a vehicle TMS with distributed TES and demonstrate it in simulation \cite{pangborn_hierarchical_2020}.  In each of these examples, state estimation and measurement are neglected because the controllers are incorporated into the simulation model, and all system states are known.  However, in practice, real-time controllers for TES require accurate state estimation, including real-time estimation of the internal temperatures and the state of charge (SOC) \cite{barz_state_2018, zsembinszki_evaluation_2020}. Experimental testing of control strategies for fast-timescale hybrid TMS, such as vehicle TMS, is limited; therefore, development and experimental validation of real-time model-based state estimation techniques for latent TES devices will contribute to filling this gap in the literature.   

For latent TES containing phase-change materials (PCM), the SOC is typically defined as a function of either (i) the fraction of PCM in either the liquid or solid phase (phase fraction) or (ii) the amount of energy stored \cite{zsembinszki_evaluation_2020}.  Methods for direct measurement of the phase fraction definition have been developed for various types of TES \cite{zsembinszki_evaluation_2020, paberit_detecting_2016}, but direct measurement of the stored energy is difficult when the stored energy can be in the form of \emph{both} latent and sensible heat. Instead, methods for determining this SOC definition must rely on indirect measurements and model-based estimators \cite{beyne_estimating_2022}.  Dynamic state estimators such as the Kalman filter or Luenberger observer are commonly used to estimate unmeasured or unmeasurable states in dynamic systems given limited measurements \cite{kalman_new_1960, luenberger_observers_1966}.  In latent TES systems, temperature and SOC estimation is particularly difficult because the temperature of the PCM is nearly constant during the phase change process when the PCM exchanges most of its energy.  In addition, the nonlinear or hybrid dynamics associated with the phase change increase the complexity of the dynamic model, making observability, stability, and convergence difficult to guarantee in advance.

\subsection{Related Work}
In experimental environments, state of charge measurement or estimation methods vary extensively and are often application-specific. Zsembinszki et al.\ present an overview of various measurement methods for determining the phase fraction \cite{zsembinszki_evaluation_2020}.  These include (i) displacement or level sensors for measuring the volume change as the PCM melts and solidifies \cite{henze_experimental_2005}, (ii) pressure sensors for measuring the volume change in enclosed PCM containers \cite{paberit_detecting_2016, steinmaurer_development_2014}, (iii) digital cameras or image sensors for locating the phase change front \cite{charvat_visual_2017}, and (iv) electrical conductivity sensors for determining the phase at selected locations in electrically conductive PCMs \cite{paberit_detecting_2016, ezan_ice_2011}.  Locating the phase change front using distributed temperature sensors is possible, but calculating the stored energy from temperature measurements introduces large uncertainties because of hysteresis, undercooling, and the isothermal nature of phase change \cite{beyne_estimating_2022}.  Beyne et al.\ discuss additional methods for estimating the stored energy; one method involves measuring the heat transfer rate between the working fluid and the PCM with temperature and mass flow rate sensors in the fluid only and then integrating the heat transfer rate over time to calculate the stored energy.  This method requires no temperature sensors embedded in the PCM, but heat transfer between the TES and the surroundings must be negligible for the method to be valid \cite{beyne_estimating_2022}. 

Research on model-based dynamic state estimation for thermal energy storage is limited.  Barz et al.\ design and implement an extended Kalman filter (EKF) for temperature and SOC estimation in a shell-and-tube TES device using temperature measurements.  The authors find that the SOC estimated by the EKF tracks that of a high-fidelity simulation more closely than the SOC calculated directly from the temperature measurements \cite{barz_state_2018}. Similarly, Pernsteiner et al.\ implement an EKF for state and parameter estimation of a reduced-order PCM-based TES simulation model \cite{pernsteiner_state_2021}.  Jaccoud et al.\ investigate a particle filter, a type of Monte Carlo method for state estimation, for estimating the temperature and phase change front location in a 1-dimensional heat transfer problem \cite{jaccoud_state_2018}.  However, a limitation of these works is the absence of guarantees on the convergence of the state estimates. Although Barz et al.\ \cite{barz_state_2018} verify local observability, none of these authors attempt to prove uniform observability or uniform detectability of the nonlinear system model or guarantee convergence of the state estimates. One exception is Morales Sandoval et al.\ who utilize a nonlinear Luenberger observer for temperature estimation in a \emph{sensible} TES device, a hot water thermal storage tank.  The authors check observability of their five-state model but find that the system is only observable when all five states are measured \cite{morales_sandoval_design_2021}.  Hence, the thermal storage device must be designed with observability and state estimation in mind to ensure there are a sufficient number of sensors to fully observe the system. Increasing the number of states in a thermodynamic model would improve model accuracy, but the higher order model could be unobservable, leading to divergence of the state estimate \cite{barz_state_2018, morales_sandoval_design_2021}.  In many \emph{latent} TES architectures, including a sensor for each state in the model is not feasible because of manufacturability or cost constraints, so the estimation model must be constructed such that the system is observable, or at least detectable, given limited measurements.

\subsection{Contribution}
We fill this gap in the literature by designing and experimentally validating a state-dependent Riccati equation (SDRE) filter \cite{mracek_new_1996, jaganath_sdre-based_2005, berman_comparisons_2014} for a PCM-based thermal energy storage device integrated with a single-phase cooling loop.  In contrast to the ubiquitous nonlinear state estimators---the extended Kalman filter \cite{the_analytic_sciences_corporation_applied_1974} and the unscented Kalman filter \cite{julier_new_1997,wan_unscented_2001}---the SDRE filter uses a linear parameter-varying model parameterized by the state of the system with which uniform detectability and boundedness of the estimation error can be guaranteed in advance for many types of dynamic systems \cite{beikzadeh_exponential_2012}.  A key difference between the SDRE filter and the EKF is that the SDRE filter does not require the Jacobian matrix for linearization.  For latent TES, the highly nonlinear nature of phase change makes using estimation methods requiring linearization especially undesirable; derivation of the Jacobian matrix can be difficult, and its calculation for higher order models can be too computationally intensive for real-time estimators.  Additionally, in highly nonlinear systems, linearization can lead to poor performance, instability, and loss of observability. Fortunately, these issues can be avoided by using the SDRE filter \cite{ewing_analysis_2000}.


We leverage a finite-volume model of a PCM-based TES and limited temperature measurements to estimate the temperature distribution inside the TES using the continuous-discrete SDRE formulation \cite{berman_comparisons_2014}.  A state of charge metric is defined that is directly calculated from the estimated temperatures. We show that the nonlinear finite-volume model, when parameterized as a linear parameter-varying model for the SDRE filter, is uniformly detectable, thereby guaranteeing boundedness of the state estimation error.  Furthermore, we show that the SDRE approach can be generalized to any phase-change TES architecture by constructing the finite-volume heat transfer model with a strongly connected weighted graph, which will be uniformly detectable given at least one state measurement. \revres{ While there are other nonlinear state estimators (particle and Kalman-type) \cite{simon_optimal_2006} which, like the SDRE, do not require the Jacobian, this graph-structured model for which we can directly show detectability and error boundedness with limited sensing, makes the SDRE a natural choice.}  Through a series of simulated and experimental case studies, we demonstrate the boundedness and accuracy of the SDRE filter applied to the nonlinear estimation problem of a latent TES device. 

This paper is organized as follows.  In Section~\ref{sec:model} we present the TES architecture as well as the system model---both a higher fidelity one used for simulation and a reduced-order model used for estimation. Section~\ref{sec:SDRE} describes the SDRE filter formulation and proves that the estimation model is uniformly detectable. The experimental setup is described in Section~\ref{sec:setup}.  Section~\ref{sec:testing} presents a series of case studies demonstrating the estimator's accuracy and convergence in both simulation and experiments.

\section{TES Thermodynamic Model} \label{sec:model}

The TES device considered in this work is designed for use in a vehicle thermal management system to provide supplementary heat rejection capability during periods of large transient heat loading. The TES design consists of a paraffin phase change material, hexadecane, with a melting point of 289.5 K \cite{hale_phase_1971} embedded in a rectangular fin heat sink.  The working fluid enters the device and flows along a flat metal separator plate below the PCM layer, as shown in Fig.\ \ref{fig:TES_3D}.  The fins in the PCM layer increase the thermal conductivity of the PCM/fin composite.  The entire TES module is contained in an insulating enclosure to reduce heat transfer with the surroundings. 
\begin{figure}[tpb]
    \centering
    \includegraphics[width=3.1in]{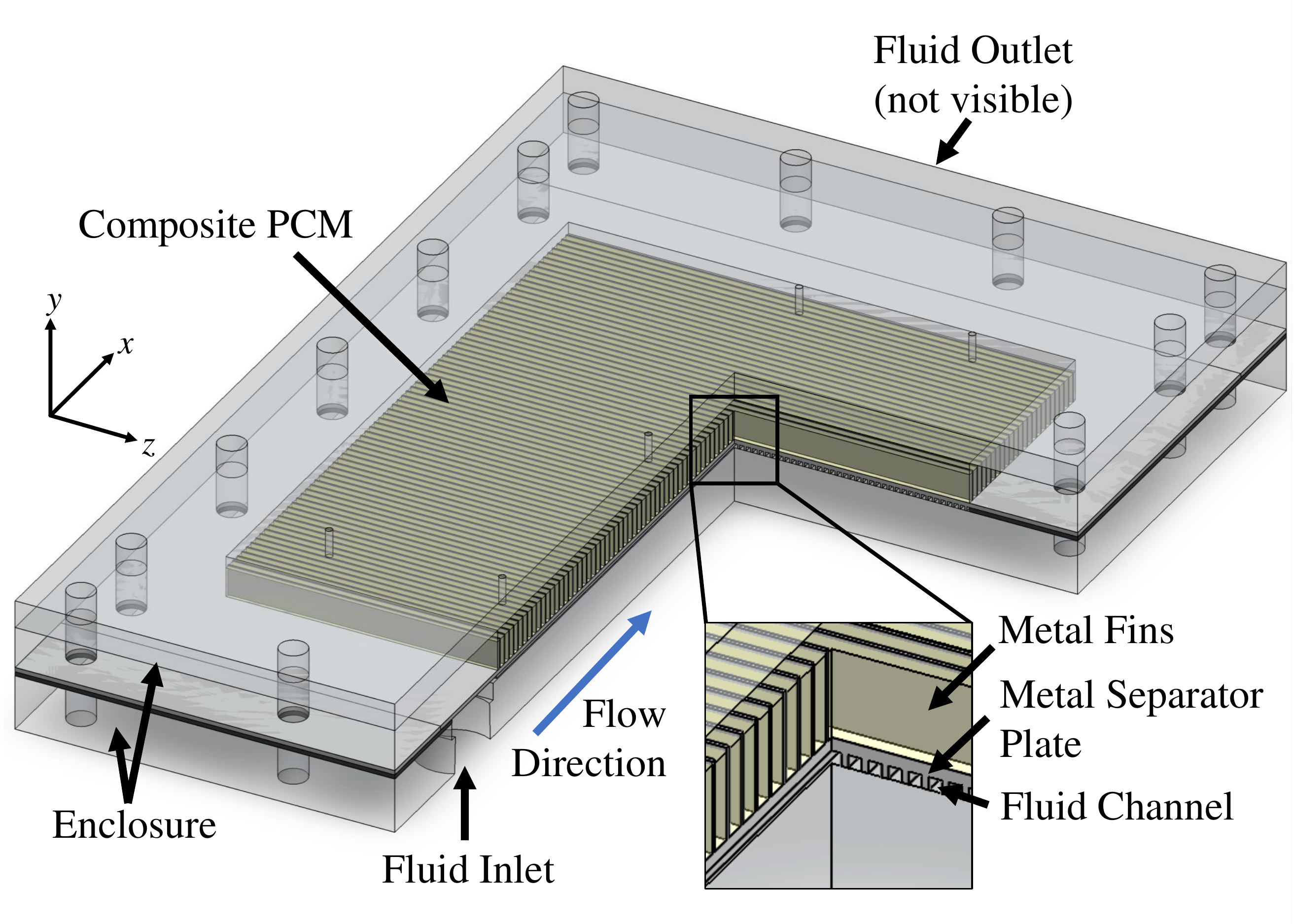}
    \caption{Section view of a 3D model of the TES module.  The fluid outlet is not visible, but the device is symmetric and the outlet geometry is identical to the inlet.}
    \label{fig:TES_3D}
\end{figure}

Experimental testing of the TES design has been used to validate a high-fidelity nonlinear simulation model, which is briefly discussed in Section~\ref{sec:fvm}.  More detail about the simulation model can be found in \cite{gohil_reduced-order_2020} and \cite{shanks_design_2022}. The TES model used for state estimation, discussed in Section~\ref{sec:SE_model}, is based on this validated simulation model.

\subsection{Finite Volume Method} \label{sec:fvm}

To model heat transfer in the TES, we derive a model based on finite-volume discretization as described in \cite{gohil_reduced-order_2020}.  Here we summarize the model for the benefit of the reader. We discretize the cross-section of the device into a grid of $n_x\times n_y$ control volumes, as shown in Fig.\ \ref{fig:TES_diag}, and derive an energy balance equation for each control volume.  The fluid channel and metal plate each comprise one layer of control volumes, and the remaining $n_y-2$ layers contain the PCM and metal fins. The set of all control volumes in the fluid channel is denoted as $\mathcal{F}$, and the fluid control volumes are numbered $j=1$ to $j=n_x$ from the inlet to the outlet\footnote{In Eqn.\ \eqref{eq:advection}, when $j=1$, the term $T_{j-1}=T_{0}$ is defined as the temperature of the fluid entering the TES, or $T_{in}$.}.  The set of metal plate control volumes, comprising $j=n_x+1$ to $j=2n_x$, is denoted as $\mathcal{P}$, and the set of PCM/fin control volumes ($j=2n_x+1$ to $j=n_xn_y$) is denoted as $\mathcal{C}$.

Eqn.\ \eqref{eq:energy_bal} defines an energy balance for control volume $j$.  Heat transfer along the width of the TES ($z$ direction) is assumed to be negligible, and temperatures across the width are assumed to be uniform.  The $\dot{Q}^{adv}_j$ term, defined in Eqn.\ \eqref{eq:advection}, represents heat transfer in fluid control volumes due to mass transfer, or advection; $\dot{m}$ is the mass flow rate, and $c_{p,f}$ is the specific heat of the working fluid.  Conductive heat transfer in the fluid channel is neglected because it is negligible compared to the advective and convective heat transfer rates.  The  $\dot{Q}_{i\rightarrow j}$ term, defined in Eqn.\ \eqref{eq:heat_transfer}, represents the conductive (solid-solid) or convective (fluid-solid) heat transfer rate from $i$ to $j$, where $i\in\mathcal{A}(j)$, and $\mathcal{A}(j)$ is the set of up to four control volumes adjacent to $j$.  Heat transfer between adjacent control volumes is modeled using the thermal resistance, given in Eqn.\ \eqref{eq:resistance}, which is either conductive or convective. In Eqn.\ \eqref{eq:resistance}, $d_{i,j}$ is the distance between the centers of control volumes $i$ and $j$, $\kappa_j$ is the thermal conductivity of $j$, $a_{i,j}$ is the area of the boundary between $i$ and $j$, and $U_j$ is the convective heat transfer coefficient of fluid control volume $j$.  
\begin{figure}[tpb]
    \centering
    \includegraphics[width=3.1in]{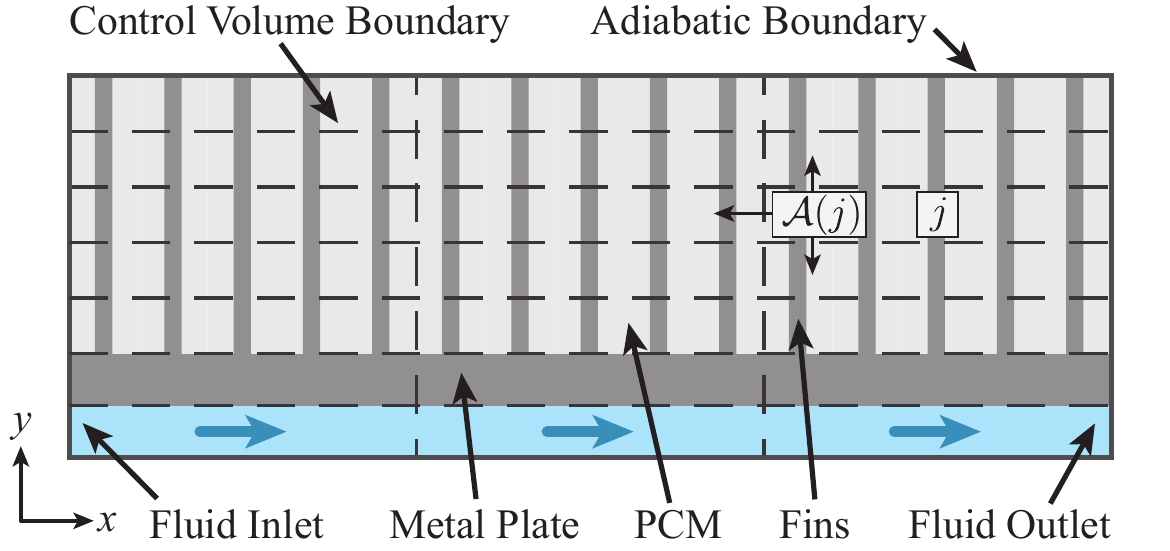}
    \caption{Cross section of the TES module considered in this work (not to scale).  A finite volume heat transfer model is constructed by discretizing the cross section into a grid of rectangular control volumes.  In this figure, the dashed lines demarcate the grid of control volumes for the model used for state estimation.}
    \label{fig:TES_diag}
\end{figure}
\begin{subequations}
\begin{gather}
    m_{j}c_{p,j} \frac{d T_j}{d t}=\dot{Q}^{adv}_j +    \sum_{i \in \mathcal{A}(j)} \hspace{-4pt}\dot{Q}_{i\rightarrow j}\label{eq:energy_bal} \allowdisplaybreaks\\
    \dot{Q}^{adv}_j = \begin{cases}
        \dot{m}c_{p,f}\left(T_{j-1} - T_j\right)&j\in\mathcal{F}\\
        0&j\in\mathcal{P}\cup\mathcal{C}
    \end{cases} \label{eq:advection}\allowdisplaybreaks\\
    \dot{Q}_{i\rightarrow j} = \frac{T_{i}-T_{j}}{R_{j,i} + R_{i,j}} \label{eq:heat_transfer} \allowdisplaybreaks\\
    R_{j, i} = 
    \begin{cases}
        \dfrac{1}{U_j a_{i,j}} & i\in\mathcal{P};\:j \in\mathcal{F} \vspace{3pt}\\
        \dfrac{d_{i, j}}{2 \kappa_j a_{i,j}}& j \in{\mathcal{P}\cup\mathcal{C}}
    \end{cases} \label{eq:resistance}
\end{gather}
\end{subequations} 

The PCM/fin layer is modeled as a composite material composed of metal fins and PCM, which we call a composite PCM or CPCM. In Eqns.\ \eqref{eq:energy_bal} and \eqref{eq:resistance}, the thermal conductivity $\kappa_j$ and specific heat $c_{p,j}$ are temperature-dependent composite properties for control volumes in the CPCM layer.  Thermal properties of the metal plate and fluid layers are assumed constant. Equations for the composite properties can be found in \cite{shamberger_cooling_2018} and \cite{tamraparni_design_2021}.  Validation of the composite assumption for closely-spaced parallel rectangular fins is given in \cite{tamraparni_design_2021}. 

The phase change dynamics are modeled using the effective specific heat function given in Eqn.\ \eqref{eq:cp_eff}, which is validated in \cite{gillis_numerical_2021} and \cite{yangSolvingHeatTransfer2010}.  Near the phase change temperature $T_{pc}$, the effective specific heat of CPCM is greatly increased so that the latent heat is modeled as sensible heat over the temperature range $T_{pc}\pm \frac{\Delta T_{pc}}{2}$, with $\Delta T_{pc}=8$ K, as shown in Fig.\ \ref{fig:sp_heat}.  The width of this temperature range is determined by the parameter $\alpha = \frac{8}{\Delta T_{pc}}$.  Differential scanning calorimetry (DSC) of hexadecane has shown that the phase change occurs over this range of temperatures, although hexadecane also exhibits undercooling during the solidification process and thermal hysteresis between the melting and solidifying temperatures \cite{sgreva_thermo-physical_2022}.  These additional phase-change phenomena are neglected in the simulation model. Outside the latent temperature range, the specific heat approaches the liquid specific heat $c_{p,liq}$ for $T_j > T_{pc}$ or the solid specific heat $c_{p,sol}$ for $T_j < T_{pc}$.  The specific enthalpy of fusion, $h_{fus}$, and the solid and liquid specific heats, $c_{p,sol}$ and $c_{p,liq}$, represent properties of the CPCM. 
\begin{multline}\label{eq:cp_eff}
c_{p,j}(T_j)
= c_{p,sol}+\frac{\left(c_{p,liq}-c_{p,sol}\right)}{1+e^{-\alpha\left(T_j-T_{pc}\right)}} \\
+\frac{h_{fus}  \alpha}{2+e^{-\alpha\left(T_j-T_{pc}\right)}+e^{\alpha\left(T_j-T_{pc}\right)}} \;\forall j\in \mathcal{C}
\end{multline} 

During phase change, a control volume will contain both solid and liquid CPCM. The partially melted CPCM is then treated as an isotropic composite material \cite{yangSolvingHeatTransfer2010}; the volume fraction of liquid CPCM is given in Eqn.\ \eqref{eq:mf} \cite{gillis_numerical_2021}.
\begin{equation}\label{eq:mf}
f_{m,j}=\frac{1}{1+e^{-\alpha\left(T_j-T_{pc}\right)}}\;\forall j\in \mathcal{C}
\end{equation}
\begin{figure}[tpb]
    \centering
    \includegraphics[width=3.1in]{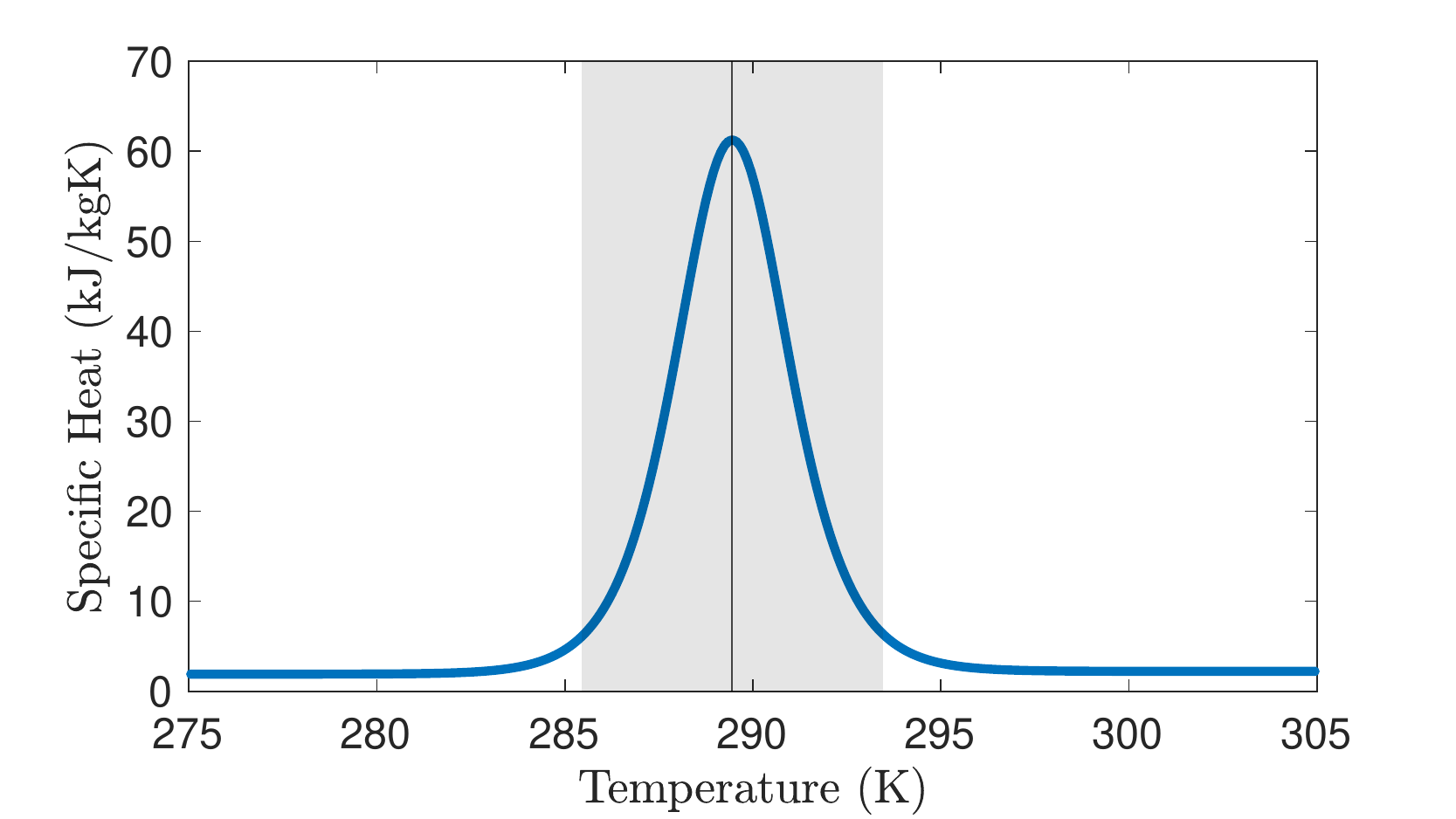}
    \caption{Effective specific heat function for simulating phase change in the TES device.  Shading represents the approximate range of temperatures over which the phase change occurs.  The black line is the true melting point, 289.5 K.}
    \label{fig:sp_heat}
\end{figure}

\subsection{Model Used for State Estimation}\label{sec:SE_model}
The number of states in the finite-volume model is a function of the level of discretization chosen by the user.  For a higher fidelity model, the user may choose $n_x$ and $n_y$ to be on the order of tens of control volumes.  However, this would results in a model that is not suitable for state estimation.  Instead, for the purpose of state estimation, we choose a grid of $n_x=3$ by $n_y=7$ control volumes for a total of $n=21$ states. Defining an energy balance for each control volume yields a system of differential equations that can be represented using a linear parameter-varying state-space model.  This model can be quickly derived using graph-based methods by representing the finite-volume model as the thermal resistance network in Fig.\ \ref{fig:res_network}.  

\begin{proposition}\label{prop:graph} The thermal resistance network in Fig.\ \ref{fig:res_network} forms a connected undirected weighted graph. The control volumes constitute the vertices of the graph, and the thermal resistances between control volumes are the graph's edges.\end{proposition}
\begin{proof}
If a path, or a set of contiguous edges, exists between any two vertices, an undirected graph is said to be connected. Proposition \ref{prop:graph} follows from this definition.
\end{proof}
The state-space model in Eqn.\ \eqref{eq:TES_ss} is derived by calculating the graph's Laplacian matrix, $L(x)$, in addition to a diagonal capacitance matrix, $M(x)$, and an input matrix, $B(x)$, where the state vector
$x=\begin{bmatrix}
T_1 & \cdots & T_n
\end{bmatrix}^\top\in\mathbb{R}^n$ contains the temperatures of all $n$ control volumes; this includes the fluid control volumes, metal plate control volumes, and CPCM control volumes. 
The notation $x_k$ refers to the value of $x$ at time $t_k$, or $x_k = x(t_k)$. The control input $u\in \mathbb{R}^1$ is the mass flow rate, and $y_k\in \mathbb{R}^p$ is the (output) vector of states that are measured at time step $k$.  
\begin{subequations}\label{eq:TES_ss}
\begin{align}
\begin{split}
    \dot{x}&=-M^{-1}(x)L(x) x + B(x)u 
    \\ &= A(x)x + B(x)u 
\end{split}\\
    y_k &= Cx_k \label{eq:output}
\end{align}
\end{subequations}
\begin{figure}[tpb]
    \centering
    \includegraphics[width=3.1in]{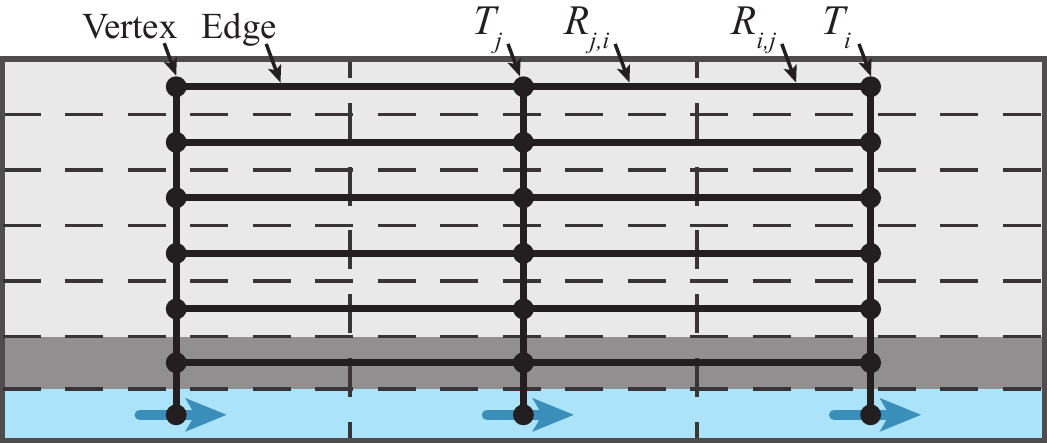}
    \caption{The finite-volume model forms a resistance network and can be analyzed as a connected graph. Each edge consists of two series resistances, one for each control volume connected by the edge.  The fluid control volumes are not connected by edges because conduction in the fluid is neglected. }
    \label{fig:res_network}
\end{figure}

The Laplacian matrix $L(x)\in\mathbb{R}^{n\times n}$ for the weighted graph is defined in Eqn.\ \eqref{eq:laplacian}\footnote{The Laplacian of a graph is more formally defined as the difference between the degree matrix and the adjacency matrix; see \cite{inyang-udoh_strongly_2022}}. The matrix $M(x)\in\mathbb{R}^{n\times n}$ is defined in Eqn.\ \eqref{eq:cap_mat} and contains the thermal capacitance $m_jc_{p,j}$ of each control volume; for CPCM control volumes, the specific heat is defined using Eqn.\ \eqref{eq:cp_eff}.  The input matrix $B(x)\in \mathbb{R}^{n\times1}$ is defined in Eqn.\ \eqref{eq:input_mat} and contains the advection terms. Finally, $T_{in}$ is the temperature of the fluid at the inlet, which is not a state. Recall that $T_1$, $T_2$, \dots, $T_{n_x}$ are the fluid control volume temperatures, which are the first $n_x$ states in $x$.  

\begin{subequations}
\begin{gather}
    L_{i,j} = 
    \begin{cases}
    \dfrac{-1}{R_{j,i} + R_{i,j}} &  i\neq j\textup{ and } \exists\textup{ edge } (i,j) \\
    0& i\neq j\textup{ and }\nexists\textup{  edge } (i,j) \\
    \displaystyle\underset{l\in\mathcal{A}(j)}{\sum}\hspace{-4pt}\left(-L_{l,j}\right)&  i=j 
\end{cases}\label{eq:laplacian} \allowdisplaybreaks\\
    M(x) = \begin{bmatrix}
    m_1c_{p,1} & 0 & \cdots  & 0 \\
    0 & m_2c_{p,2}   & \cdots  & 0  \\
    \vdots & \vdots & \ddots & \vdots \\
    0 & 0 & \cdots & m_nc_{p,n}  
\end{bmatrix}\label{eq:cap_mat}\allowdisplaybreaks\\
B (x)= c_{p,f}M^{-1}\begin{bmatrix}
  T_{in} - T_1 \\
  T_1 - T_2\\
   \vdots  \\
  T_{n_x-1} - T_{n_x}\\
  0 \\
   \vdots  \\
  0 
\end{bmatrix}\in\mathbb{R}^{n\times1}\label{eq:input_mat}
\end{gather}
\end{subequations}

Suppose measurements of the temperature of certain control volumes are available at discrete time instances $t_k$, $k=1,2,\dots$; then  $C\in\mathbb{R}^{p\times n}$ is a binary output matrix where $C_{i,j} = 1$ if $y_{k,i} = x_{k,j}$, and $C_{i,j} = 0$ otherwise.  Thus, the output equation (Eqn.\ \ref{eq:output}) is discrete-time although the state dynamics are continuous.

\subsection{State of Charge}
We define the state of charge (SOC) of the TES in Eqn.\ \eqref{eq:soc} as a piecewise linear function of the energy stored in the device, or the enthalpy of the CPCM, $H$ \cite{shanks_control_2022}.  Following the convention described in \cite{zsembinszki_evaluation_2020}, minimum stored energy corresponds to a maximum SOC of 1, and maximum stored energy corresponds to a minimum SOC of 0. Note that this definition of SOC considers both latent and sensible energy storage within the TES.
\begin{equation}
x_{soc}=\begin{cases}
1 & H<H_{min } \\
\dfrac{H_{max}-H}{H_{max }-H_{min }} & H_{min } \leq H \leq H_{max } \\
0 & H>H_{max}
\end{cases}
\label{eq:soc}
\end{equation}

With the finite volume model, the total stored energy, given in Eqn.\ \eqref{eq:total_enth}, is the sum of the enthalpies of each CPCM control volume.  The specific enthalpy of the PCM/fin composite in Eqn.\ \eqref{eq:abs_enth} is found by integrating the effective specific heat function, Eqn.\ \eqref{eq:cp_eff}, with respect to temperature.  Fig.\ \ref{fig:sp_enth} shows the shape of the specific enthalpy curve; the specific enthalpy is defined such that it is zero at the phase change temperature, $T_{pc}$. The maximum enthalpy $H_{max}$ and minimum enthalpy $H_{min}$ in Eqn.\ \eqref{eq:soc} are calculated for the PCM at a uniform temperature $T_{max} = 308$ K or $T_{min} = 278$ K, which are user-defined parameters corresponding to the maximum and minimum expected temperatures, respectively.  Thus, the SOC of the TES can be calculated given the temperature of each CPCM control volume.
\begin{gather}
\begin{multlined}
\label{eq:total_enth}
    H=\sum_{j\in\mathcal{C}}m_{j} h(T_{j})
    \end{multlined}\allowdisplaybreaks\\
\begin{multlined}\label{eq:abs_enth}
    h(T_{j\in\mathcal{C}})= \int_{T_{pc}}^{T_j}c_{p,j}(x)dx \\ 
    = \frac{h_{fus}}{2} \tanh \left[\frac{\alpha}{2}\left(T_j-T_{pc}\right)\right] + 
    (T_j-T_{pc})c_{p,sol} \\
    +\frac{c_{p,liq}-c_{p,sol}}{\alpha}\ln{\left[\frac{1+e^{\alpha\left(T_j-T_{pc}\right)}}{2}\right]} 
\end{multlined}
\end{gather}
\begin{figure}[tpb]
    \centering
    \includegraphics[width=3.1in]{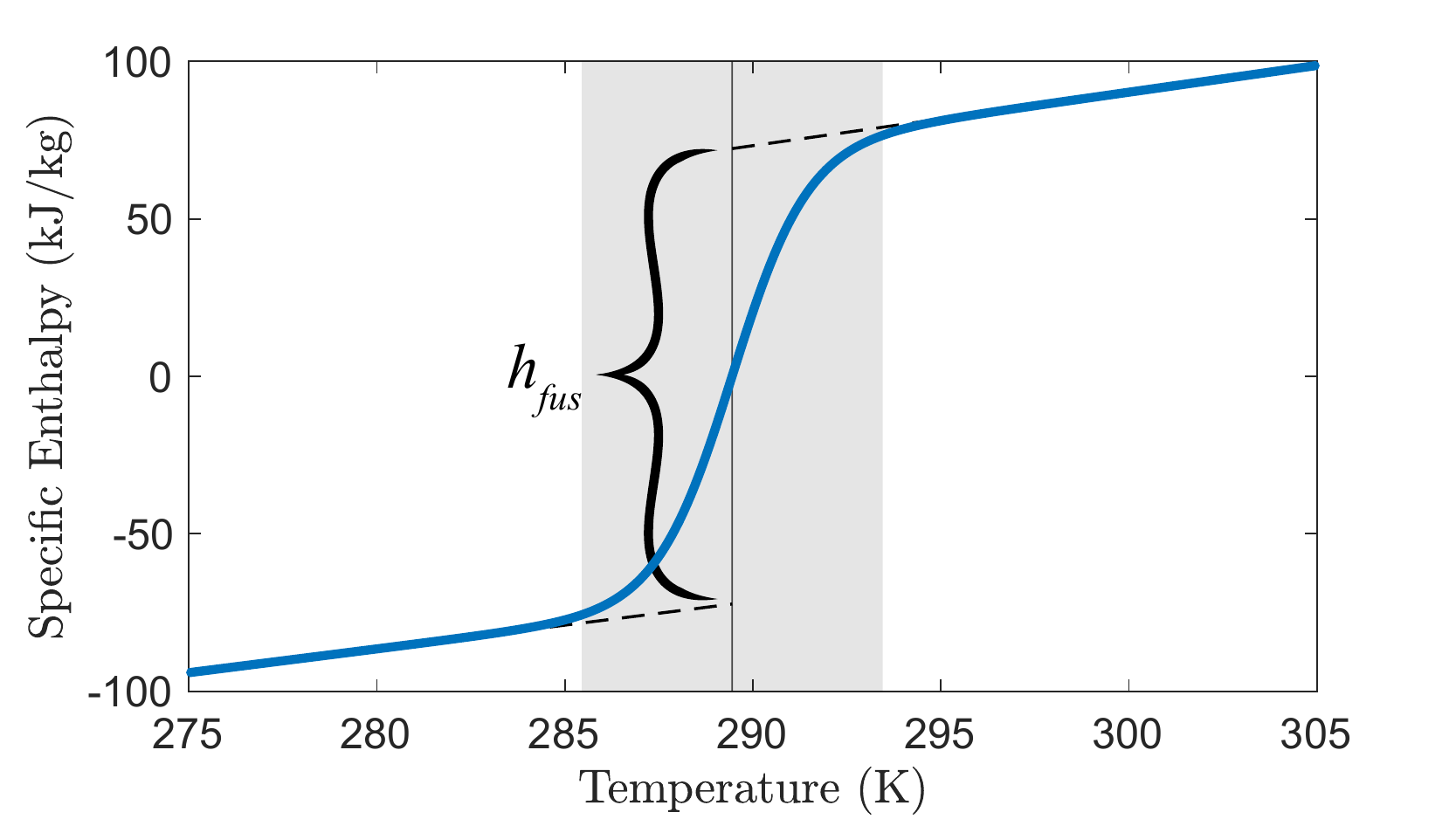}
    \caption{Specific enthalpy function for the CPCM derived by integrating the effective specific heat function.  Shading represents the approximate range of temperatures over which the phase change occurs.  The black line is the phase change temperature, 289.5 K.}
    \label{fig:sp_enth}
\end{figure}
\section{State-Dependent Riccati Equation Filter} \label{sec:SDRE}

In \cite{mracek_new_1996}, Mracek et al.\ discuss the continuous-time formulation of the SDRE filter in which the observer gain is obtained by solving the continuous-time state-dependent Riccati equation \cite{mracek_new_1996}. In \cite{jaganath_sdre-based_2005}, Jaganath et al.\ present the discrete-time SDRE filter in which the discrete-time Riccati equation is solved recursively with prediction and update steps resembling the EKF algorithm.  For systems with continuous dynamics and discrete sampling, Berman et al.\ \cite{berman_comparisons_2014} and Inyang-Udoh et al.\ \cite{inyang-udoh_sampling_2022} describe the continuous-discrete SDRE filter in which the state estimates and covariance matrix are propagated in discrete time steps using a state-dependent state transition matrix. In this section we outline the continuous-discrete SDRE filter algorithm and describe its application to the latent TES state estimation problem. As we will show in Section \ref{sec:sample_rate}, the continuous-discrete formulation can be particularly advantageous when the sample rate is limited.

Let the input $u = \dot m(t)$ be constant over each interval $\Delta t_k = t_{k+1} - t_k$. Further assume some process noise ${w}_k \sim \mathcal{N}({0}, W_k)$ and measurement noise ${v}_{k} \sim \mathcal{N}({0}, V_{k})$ are present, where $W_k \in \mathbb{R}^{n\times n}$, $V_{k} \in\mathbb{R}^{p\times p}$ are covariance matrices.  If the system dynamics change little during the interval $\Delta t_k$, the linear parameter-varying state-space model can be ``frozen-in-time'' (see \cite{inyang-udoh_strongly_2022}) and discretized using the exact solution for the state transition matrix of a linear dynamic system.  The  discrete-time system is given by 
\begin{subequations}\label{eq:sys_nonlin}
    \begin{align} 
        x_{k+1} &= \Phi_kx_k + \Gamma_{k}u_k + w_k \\
        y_{k} &= Cx_{k} + v_{k},    \end{align}\label{eq:discrete_sys}
\end{subequations}
where 
\begin{subequations}
    \begin{align}
        \Phi_k &= e^{A(x_k)\Delta t_k},\\
        \Gamma_k &= \int_{0}^{ \Delta t_{k}}e^{A(x_k)\tau}B(x_k)d\tau.
\end{align}\label{eq:discrete_mat}
\end{subequations}

To estimate the unmeasured temperature states of the TES, we use the two-step recursive formulation of the SDRE filter given in \cite{jaganath_sdre-based_2005} and \cite{berman_comparisons_2014} where the predicted state $\hat{x}_{k+1|k}$ and estimation error covariance $P_{k+1|k}$ are given by
\begin{subequations}
    \begin{align}
        \hat{x}_{k+1|k} &= \Phi_k\hat{x}_k + \Gamma_{k}{u}_k,\\
        P_{k+1|k} &= \Phi_k P_{k|k}\Phi_{k}^\top + W_k,
    \end{align}\label{eq:prop}
\end{subequations}
the Kalman gain $K_{k}$ is given by
\begin{equation}
    K_{k} = P_{k+1|k}C^\top\left(CP_{k+1|k}C^\top+V_k\right)^{-1},
    \label{eq:kalman_gain}
\end{equation}
and the updated state, output and covariance estimates are 
\begin{subequations}
    \begin{align}
        \hat{x}_{k+1|k+1} &= \hat{x}_{k+1|k} + K_k\left(y_k-\hat{y}_k\right),\\ 
        \hat{y}_k &= C\hat{x}_k, \\
        P_{k+1|k+1} &= \left(I-K_kC\right)P_{k+1|k}.
    \end{align} \label{eq:update}
\end{subequations}

Using a smaller time interval $\Delta t_k$ in the discrete-time model results in a better approximation of the continuous nonlinear dynamics.  If measurements are not available at every time step $t_k$, the prediction step in Eqn.\ \eqref{eq:prop} can be performed multiple times between each update step; we call this method the \textit{continuous-discrete} SDRE filter \cite{berman_comparisons_2014, inyang-udoh_sampling_2022} although it still uses a discrete approximation of the continuous dynamics.  

Now we establish the boundedness of the SDRE filter when applied to the latent TES under consideration.   

\begin{proposition}\label{prop:bounded} The estimation error covariance of the SDRE filter for the TES model given by Eqn.\ \eqref{eq:TES_ss} will remain bounded.
\end{proposition}

\begin{proof}
For the discrete-time system of Eqn.\ \eqref{eq:discrete_sys}, the SDRE error covariance is bounded if the pair $\left(\Phi_k, C\right)$ is uniformly detectable \cite{anderson_detectability_1981}. The TES model given by Eqn.\ \ref{eq:TES_ss} represents a \revres{continuously connected undirected graph since $A_{i,j}(x) >0 ~\forall t$ for each existing edge  $(i,j)$ where $i\neq j $. If a parameter-varying undirected graph is continuously connected $\forall t$, the graph implements consensus, and is, hence, uniformly detectable from any node (see the Appendix). Furthermore, the pair $\left(\Phi_k, C\right)$ is uniformly detectable for any $C$ with at least
one non-zero row sum (see Lemma 3 in the Appendix).} Therefore, it follows that the error covariance of the state estimate will remain bounded.
\end{proof}

\begin{remark}
\textup{Proposition \ref{prop:bounded} can be generalized to show that any thermal energy storage device that can be spatially discretized into a contiguous lattice of control volumes (that is, the thermal resistance between any two control volumes is finite) is uniformly detectable given that at least one state measurement exists.  In other words, when the SDRE filter is used for state estimation, boundedness of the state estimate error covariance is guaranteed.}
\end{remark}

\section{Experimental Testbed} \label{sec:setup}
The TES is integrated into the experimental single-phase thermal-fluid loop testbed shown in Fig.\ \ref{fig:testbed_photo}.  The testbed has four TES modules arranged in series, each with a capacity of approximately 25.6 kJ.  
Each module uses a separate state estimator with a model tuned to the specific module's parameters, but in this work, we will show validation of the state estimator on only the first module in the series. 
\begin{figure}[tbp]
    \centering
    \begin{subfigure}{\columnwidth}
        \includegraphics[width=3.1in]{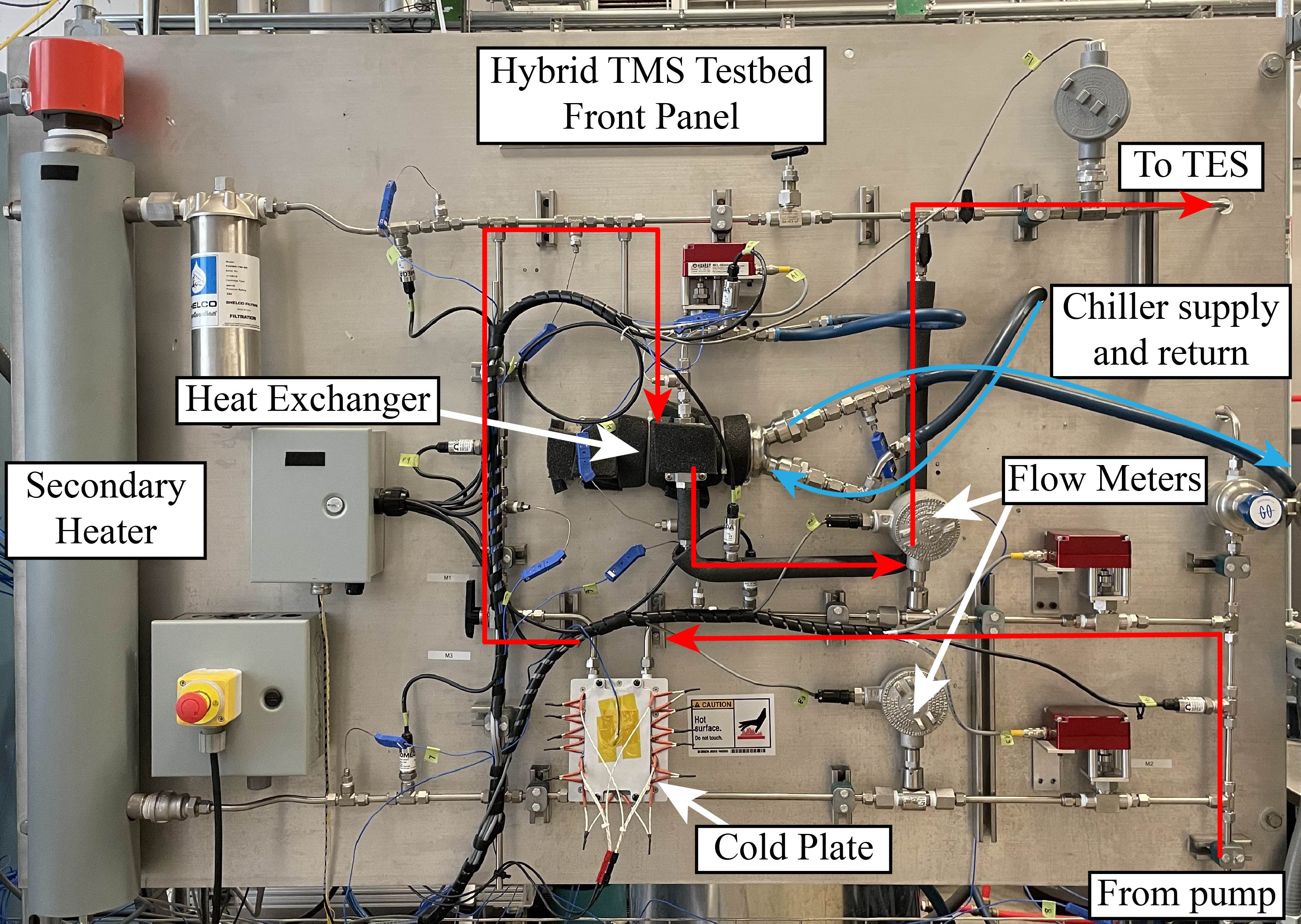}
        \caption{}
        \label{fig:front_panel}
    \end{subfigure}
    
    \begin{subfigure}{\columnwidth}
        \includegraphics[width=3.1in]{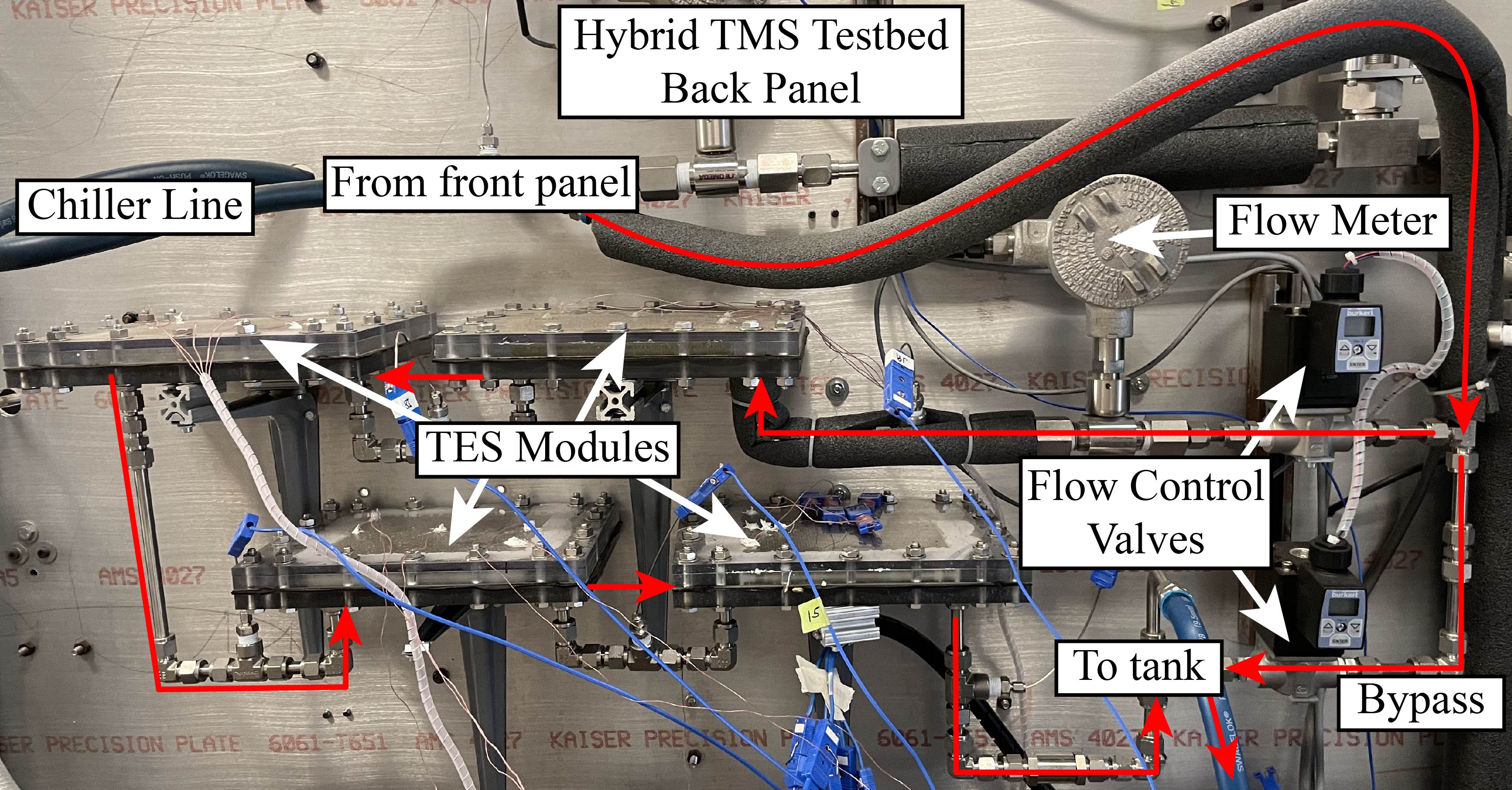}
        \caption{}
        \label{fig:back_panel}
    \end{subfigure}
    \caption{Images of the experimental single-phase thermal-fluid testbed with thermal energy storage.}
    \label{fig:testbed_photo}
\end{figure}

Figure \ref{fig:tms_diag} shows a simplified diagram of the thermal management system; a pump circulates the working fluid (water) though the primary fluid loop with a variable flow rate up to 0.183 kg/s.  Heat is added through a 6 kW 600 VDC electrical resistive heater mounted to an Advanced Thermal Solutions ATS-CP-1002 cold plate. The primary mode of heat rejection is a shell-and-tube heat exchanger cooled by a secondary chilled water loop from a Neslab HX300-DD 10 kW chiller.  A pair of Burkert Type 2875 variable position solenoid valves control the mass flow rate through the TES. The four modules are custom designed and fabricated for the testbed, and together they provide up to 2 kW of heat rejection. 
\begin{figure}[tbp]
    \centering
    \includegraphics[width=3.1in]{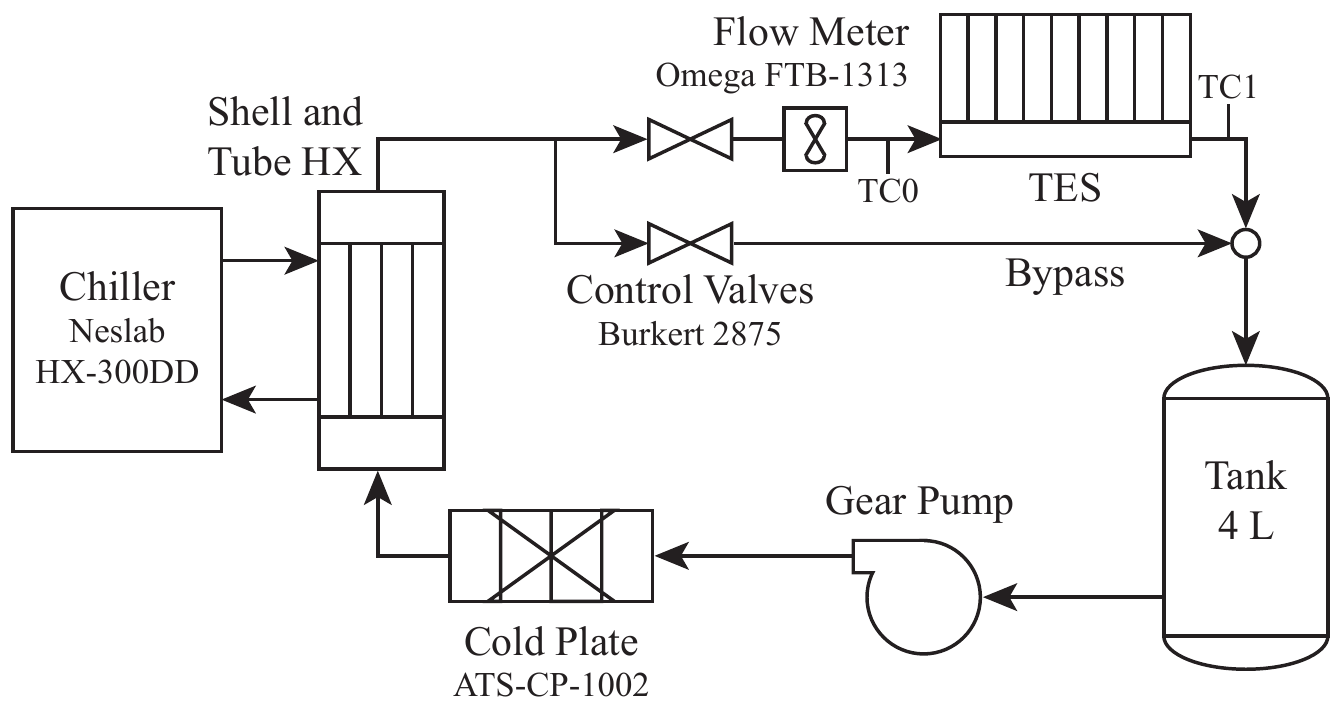}
    \caption{Simplified schematic of the experimental thermal management system.}
    \label{fig:tms_diag}
\end{figure}

Temperature measurements of the PCM layer in the TES are provided by five type T thermocouples bonded to the PCM side of the plate, labeled TC2a, TC2b, TC3, TC4a, and TC4b, in the pattern shown in Fig.\ \ref{fig:TES_3D_labeled}. TC2a and TC2b are placed at the same distance along the length of the module, and likewise for TC4a and TC4b.  Type T thermocouples labeled TC0 and TC1 in the fluid channel measure the inlet and outlet fluid temperature, respectively.  Not shown in Fig.\ \ref{fig:TES_3D_labeled} is an Omega FTB-1313 turbine flow meter upstream of the TES that measures the mass flow rate of the working fluid.
\begin{figure}[tbp]
    \centering
    \includegraphics[width=3.1in]{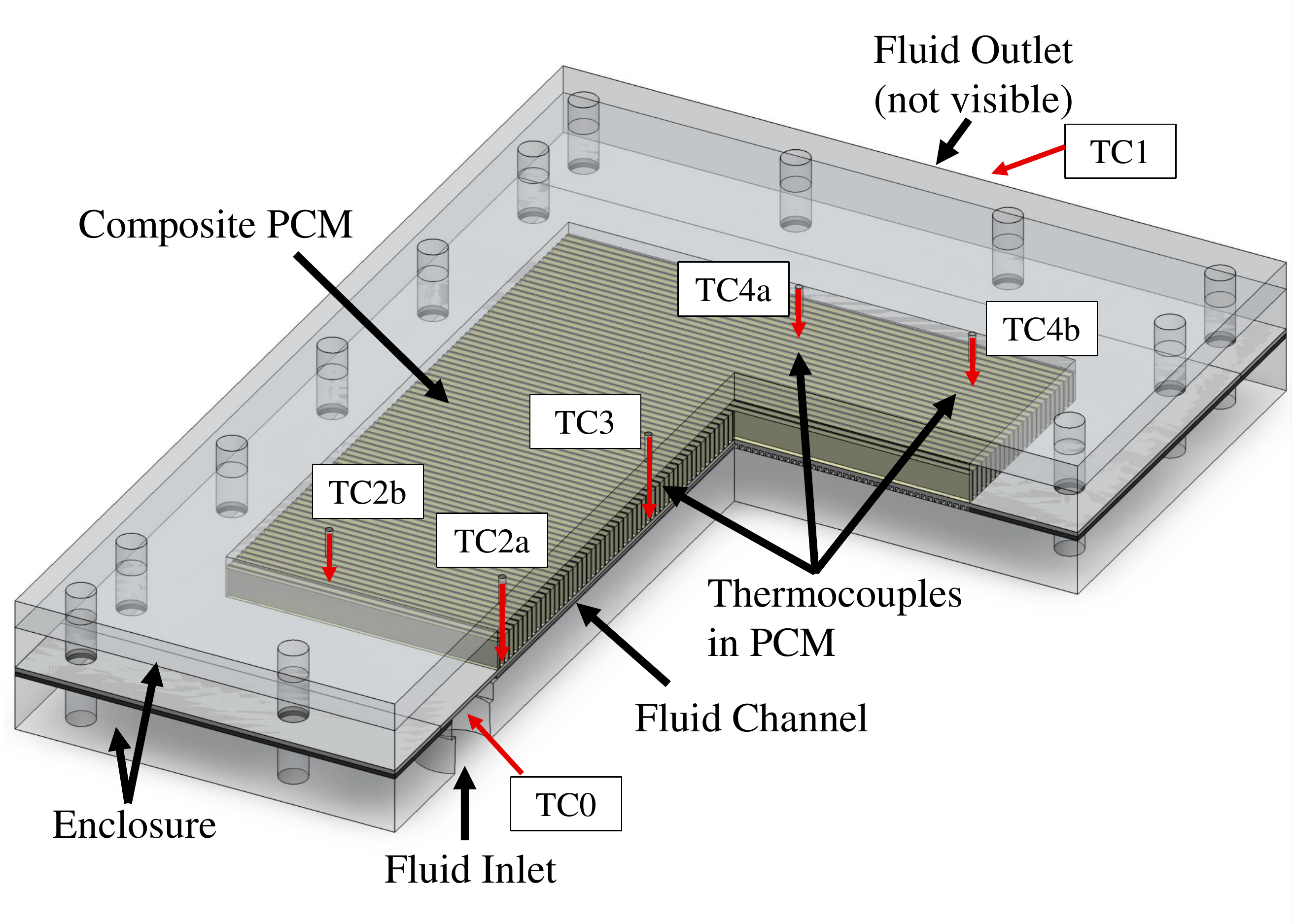}
    \caption{Section view of a 3D model of the TES module.  Thermocouples TC0 and TC1 measure the fluid inlet and outlet temperatures, respectively, and TC2a, TC2b, TC3, TC4a, and TC4b are embedded in the PCM.}
    \label{fig:TES_3D_labeled}
\end{figure}

The thermocouples are sampled at regular intervals to produce the measurement output vector $y_k$.  Thermocouple TC0 measures $T_{in}$, which is not a state but a parameter in the input matrix $B(x)$. Thermocouple TC1 measures the output from control volume CV1; the fluid outlet temperature is assumed to be equal to the temperature of the fluid in CV1.  The five thermocouples in the PCM layer lie within three control volumes (see Fig.\ \ref{fig:TC_CV_labels}). Thermocouples TC2a and TC2b lie within the same control volume, so measurements from these thermocouples are averaged to provide one output measurement, TC2, for the control volume labeled CV2. Thermocouple TC3 measures the output for control volume CV3. Measurements from thermocouples TC4a and TC4b are also averaged to provide a single output measurement, TC4, for control volume CV4. 
\begin{figure}[tpb]
    \centering
    \includegraphics[width=3.1in]{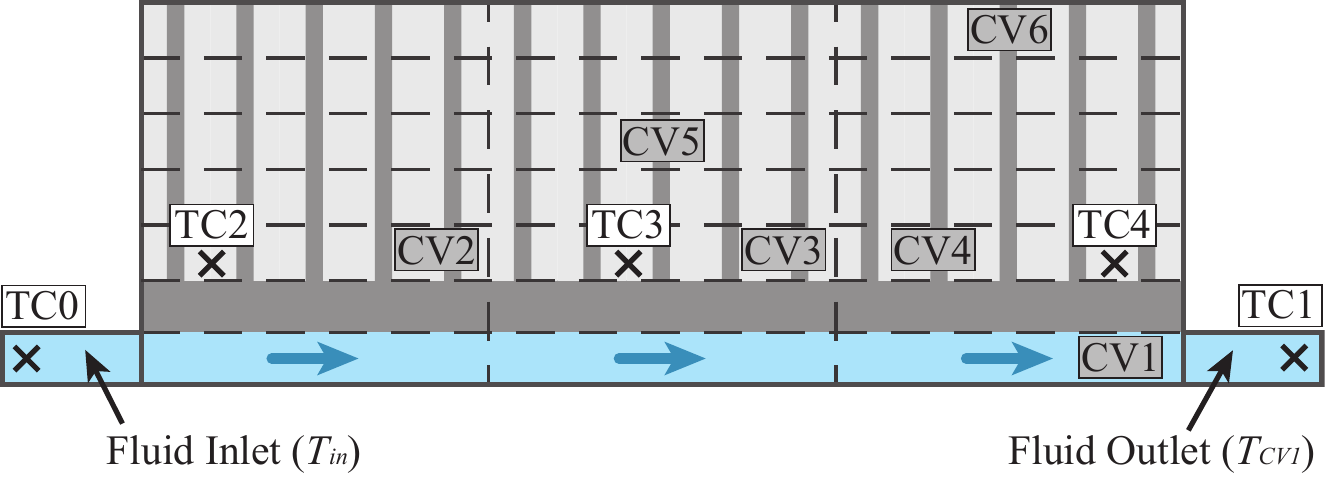}
    \caption{Locations of thermocouples TC0-4 relative to the control volume grid, denoted by $\times$ marks. Selected control volumes used for validation are also labeled CV1-6.}
    \label{fig:TC_CV_labels}
\end{figure}

During experimental tests, the SDRE filter is deployed on a National Instruments PXIe-8820 embedded controller with an Intel Celeron 1020E 2.20 GHz dual-core processor.  Data aquisition peripherals include a NI PXI-6225 16-bit analog input module for reading the flow meter and a NI PXIe-4353 24-bit thermocouple module.  Additionally, experimental data is recorded so the state estimator can be simulated offline for validation.  A custom LabVIEW VI manages the data acquisition, executes the state estimator algorithm, and translates user input into voltage signals to operate the valves, heaters, and pump.


\revres{The thermocouple sensor noise variance is determined by repeatedly sampling a nearly constant temperature, fitting a linear regression model to the sampled data, and calculating the residual variance.}  Each thermocouple is assumed to have independent Gaussian white noise.  \revres{The sample variance of each thermocouple is calculated from a dataset containing 4000 samples, which yields a median value of $\sigma^2=0.0068 K^2$. Given the similarity in noise variance across the thermocouples (which is expected given that each is identical and purchased from the same manufacturer), a single value of $\sigma^2=0.007 K^2$ (to the nearest thousandth) is assumed for each thermocouple.}  Averaging the measurements from laterally adjacent thermocouples reduces the variance of these outputs to $\text{Var}\!\left[\frac{y_{a} + y_{b}}{2}\right]=\frac{\sigma^2}{2}=0.0035\,\textup{K}^2$.  Assuming the output is $y_k = \begin{bmatrix}y_{k,TC1} & y_{k,TC2} & y_{k,TC3} & y_{k,TC4}\end{bmatrix}^\top$, the measurement noise covariance matrix is given in Eqn.\ \eqref{eq:meas_noise}.  The process noise covariance is not as easily characterized, so we assume it is negligible; Eqn.\ \eqref{eq:proc_noise} defines the matrix $W_k$ just large enough to ensure that $P_{k+1|k}$ remains positive definite.  
\begin{subequations}
\begin{gather}
    V_k = \begin{bmatrix}
        0.007 & 0 & 0 & 0\\ 0 & 0.0035 & 0 & 0\\0 & 0 & 0.007 & 0 \\ 0 & 0 & 0 & 0.0035
    \end{bmatrix}\label{eq:meas_noise}\allowdisplaybreaks\\
    W_k = 10^{-7}I_{n}\label{eq:proc_noise}
\end{gather}
\end{subequations}

\section{Testing and Validation}\label{sec:testing}

The objective of the following case studies is to validate the SDRE filter by demonstrating that the errors in the estimated temperatures and state of charge converge near zero and remain bounded.  As it is not possible to validate the SOC estimate experimentally with only five thermocouples embedded in PCM layer of the TES module, we first test the estimator using the high-fidelity simulation model validated in \cite{gohil_reduced-order_2020} and \cite{shanks_design_2022}. However, in order to enable a comparison with the experimental validation and draw conclusions about those estimation errors that are due to modeling error, we use experimental data as inputs to the simulation model, as discussed in Section \ref{sec:exp_data}. The simulation-based case study in Section \ref{sec:sim_testing} is followed by validation of the estimator's convergence and boundedness on the experimental testbed as discussed in Section \ref{sec:exp_testing}. In Section \ref{sec:sample_rate}, we investigate the effect of increased and reduced sample rates by comparing 80 samples per second (S/s), 10 S/s, 1 S/s, and 0.2 S/s.  Although the SDRE filter is capable of real-time estimation on the testbed, we perform experimental validation offline so that multiple filter configurations can be tested using the same dataset.

\subsection{Experimental Data Collection}\label{sec:exp_data}

We first collect a dataset containing temperature measurements from the seven thermocouples and a measurement of the mass flow rate on the experimental testbed at the maximum sample rate of 80 S/s.  The experimentally-measured mass flow rate and fluid inlet temperature are used as inputs for the simulation-based testing described in Section \ref{sec:sim_testing}, whereas the entire experimental dataset is used for experimental validation of the estimator in Section \ref{sec:exp_testing}. The dataset is downsampled when testing the SDRE filter with slower sample rates.  The mass flow rate through the TES is controlled by manually opening and closing the flow control valves through the LabVIEW VI. Fig.\ \ref{fig:mass_flow_rate} shows the measured mass flow rate through the TES.  Fig.\ \ref{fig:fluid_in_out} shows the temperatures of the fluid entering and leaving the TES measured by TC0 and TC1, respectively. Note that periods of zero flow rate were intentionally introduced as it would be reasonable for the TES to only be charged or discharged intermittently.  Conductive heat transfer still occurs inside the TES even when the mass flow rate is zero, so the state estimator must still work during these periods to track the changing temperature distribution.  In the subsequent figures, blue shading denotes periods of zero mass flow rate when this information is relevant to the result presented in the figure.
\begin{figure}[tpb]
    \includegraphics[width=3.1in]{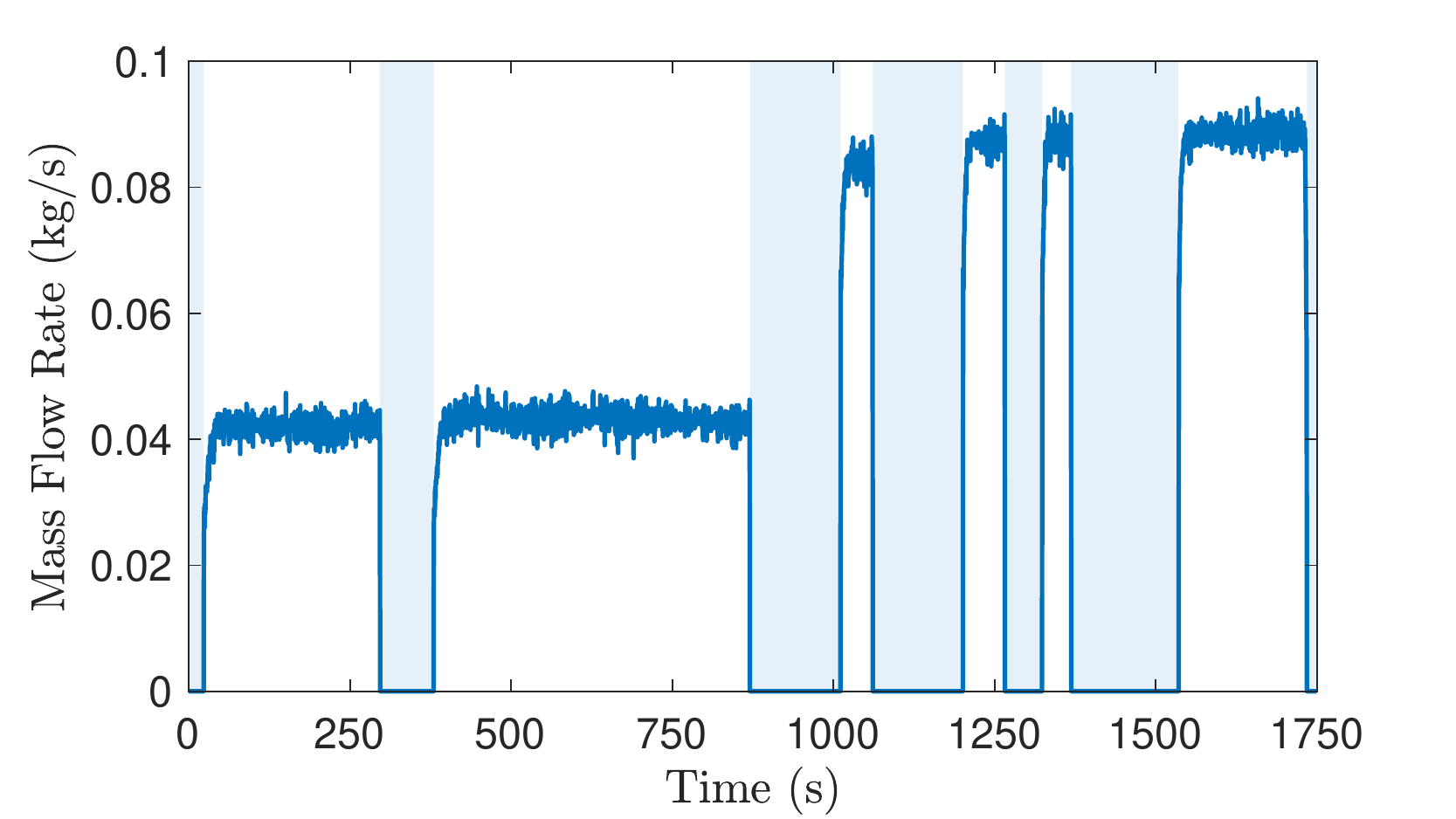}
    \caption{Measured mass flow rate through the TES. Blue shading in this figure and subsequent figures highlights periods of zero mass flow rate.}
    \label{fig:mass_flow_rate}
\end{figure}
\begin{figure}[tpb]
    \includegraphics[width=3.1in]{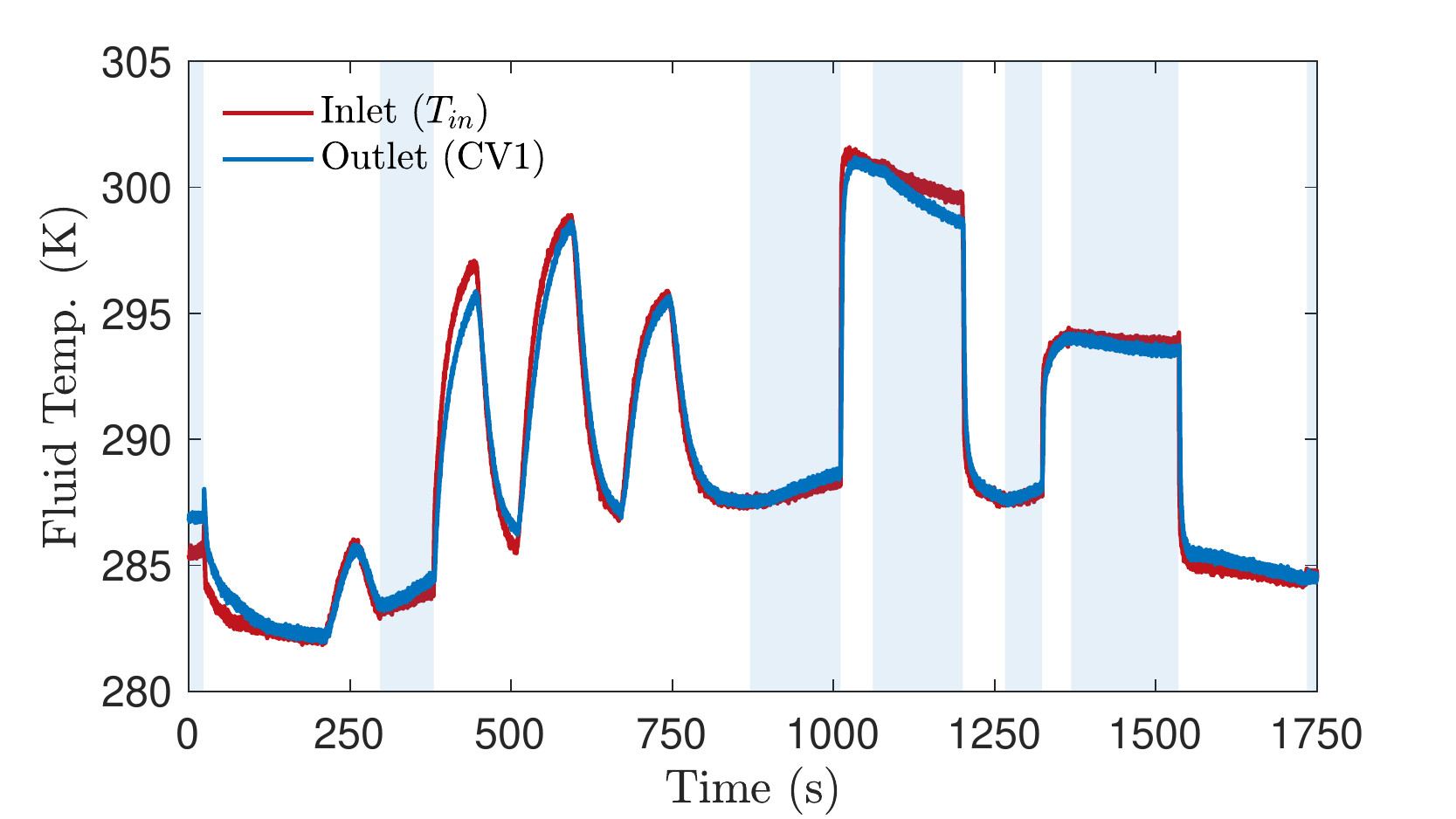}
    \caption{Inlet and outlet fluid temperature measurements.}
    \label{fig:fluid_in_out}
\end{figure}

In Section \ref{sec:sample_rate}, we investigate the effect of sample rate on the accuracy of the estimator. To make an unbiased comparison between different sample rates, the prediction equations given in Eqn.\ \eqref{eq:prop} are executed with the time step $\Delta t_k = 0.0125$~s so that the state estimator model is identical for all sample rates.  The update step given in Eqns.\ \eqref{eq:kalman_gain} and \eqref{eq:update} is performed once every 1, 8, 80, or 400 time steps depending on whether the sample rate is 80 S/s, 10 S/s, 1 S/s or 0.2 S/s, respectively.  Note that the filters receiving measurements at 10 S/s, 1 S/s, and 0.2 S/s are continuous-discrete filters, but the filter receiving measurements at 80~S/s is a discrete filter since the prediction and update steps are executed at the same rate. 

\subsection{Simulation-Based Testing}\label{sec:sim_testing}

Using the experimentally-measured mass flow rate and fluid inlet temperature (TC0) as inputs, we simulate the high-fidelity simulation model \cite{gohil_reduced-order_2020,shanks_design_2022} to generate a \emph{simulated} dataset containing the SOC and temperatures of all control volumes. Note that this model has the same finite volume structure as the model used by the estimator albeit with a finer spatial discretization of $n_x = 21$ and $n_y = 22$;\revres{ the finer discretization was chosen so that the simulation model is an accurate representation of the experimental system.  In other words, differences in accuracy between the simulation and the model used for state estimation are representative of the differences between the experimental TES and the model used for state estimation.  In the model used for state estimation, low computation time is prioritized over model accuracy}. The simulation is executed in MATLAB using the \texttt{ode23} solver \cite{shampine_matlab_1997}. The simulation calculates the temperature distribution inside the TES, converts the high-fidelity temperature distribution to a set of temperatures corresponding to the state estimator model's control volumes, and generates ``thermocouple measurements'' with a sample rate of 10 S/s by adding Gaussian noise to the temperatures of those control volumes that map to the measurement locations on the experimental TES module at discrete time steps. These simulated thermocouple measurements comprise the measurement vector $y_k$ for the state estimator; the thermocouple locations and the labels of the corresponding control volumes are given in Fig.\ \ref{fig:TC_CV_labels}.  We then test the state estimator on the simulated dataset, using simulated values of the thermocouple measurements, to establish convergence of the SOC estimate. 

First, we verify that the temperature estimates of unmeasured control volumes converge toward the temperatures calculated by the simulation, and that the estimation errors remain bounded.  In this case study, the state estimator receives the full measurement output vector at a rate of 10 S/s.  Fig.\ \ref{fig:TCall_sim_a} compares the estimated temperatures of the control volumes labeled CV5 and CV6 (see Fig.\ \ref{fig:TC_CV_labels} for the locations of these control volumes) to the corresponding simulated temperatures.  The gray shading represents the range of temperatures over which the latent heat of fusion is approximated by a large increase in specific heat. Since it is difficult to distinguish the estimates from the simulated values in Fig.\ \ref{fig:TCall_sim_a}, we show the estimation errors in Fig.\ \ref{fig:TCall_sim_b}.  Note that all estimation errors are calculated as $e(t_k)=\hat{x}(t_k)-x_{sim}(t_k)$ where $x_{sim}$ is a state variable calculated by the simulation. 
\begin{figure}[tb]
    \centering
    \begin{subfigure}[b]{\columnwidth}
        \includegraphics[width=3.1in]{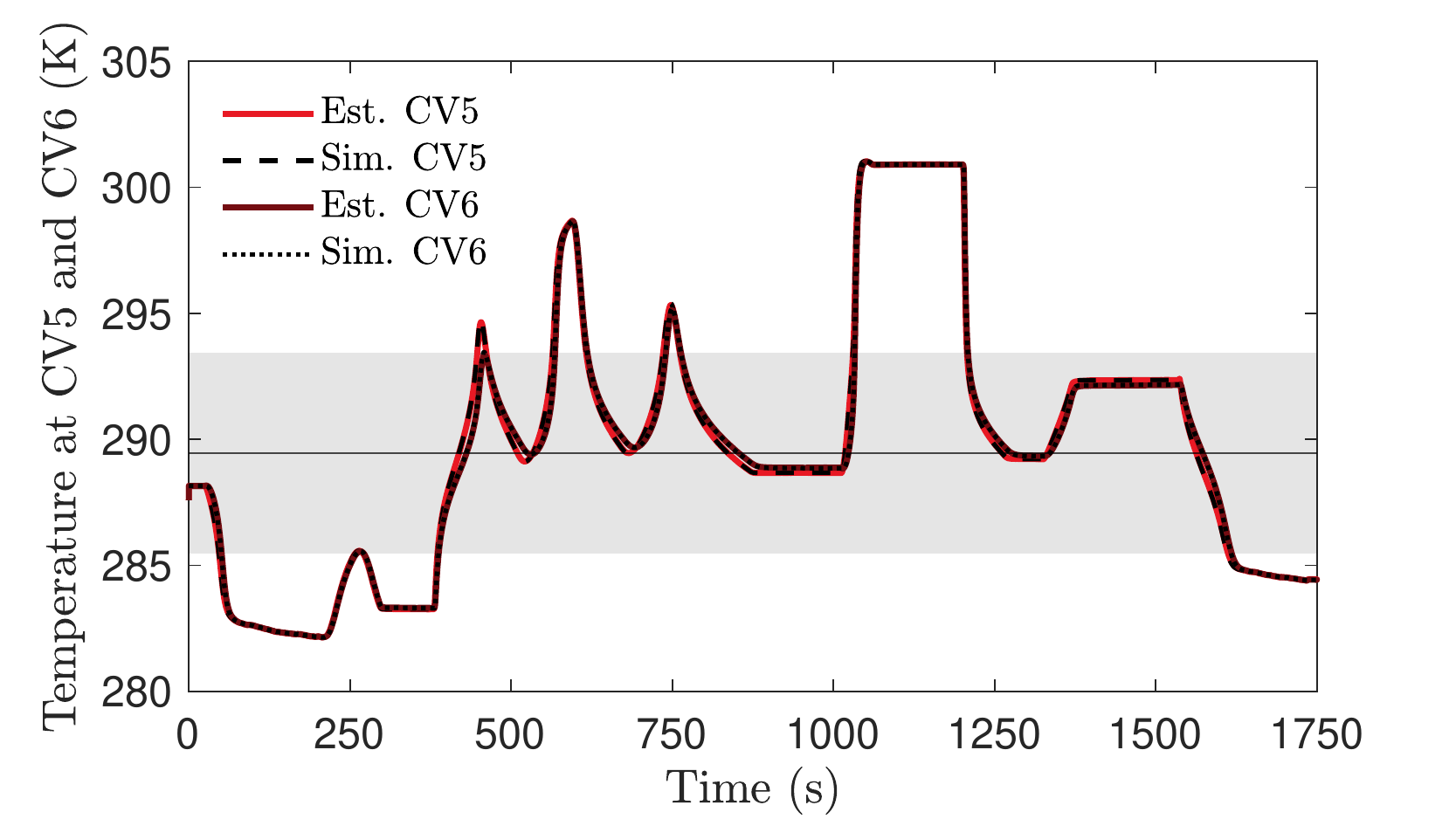}
        \caption{}
        \label{fig:TCall_sim_a}
    \end{subfigure}
    \begin{subfigure}[b]{\columnwidth}
        \includegraphics[width=3.1in]{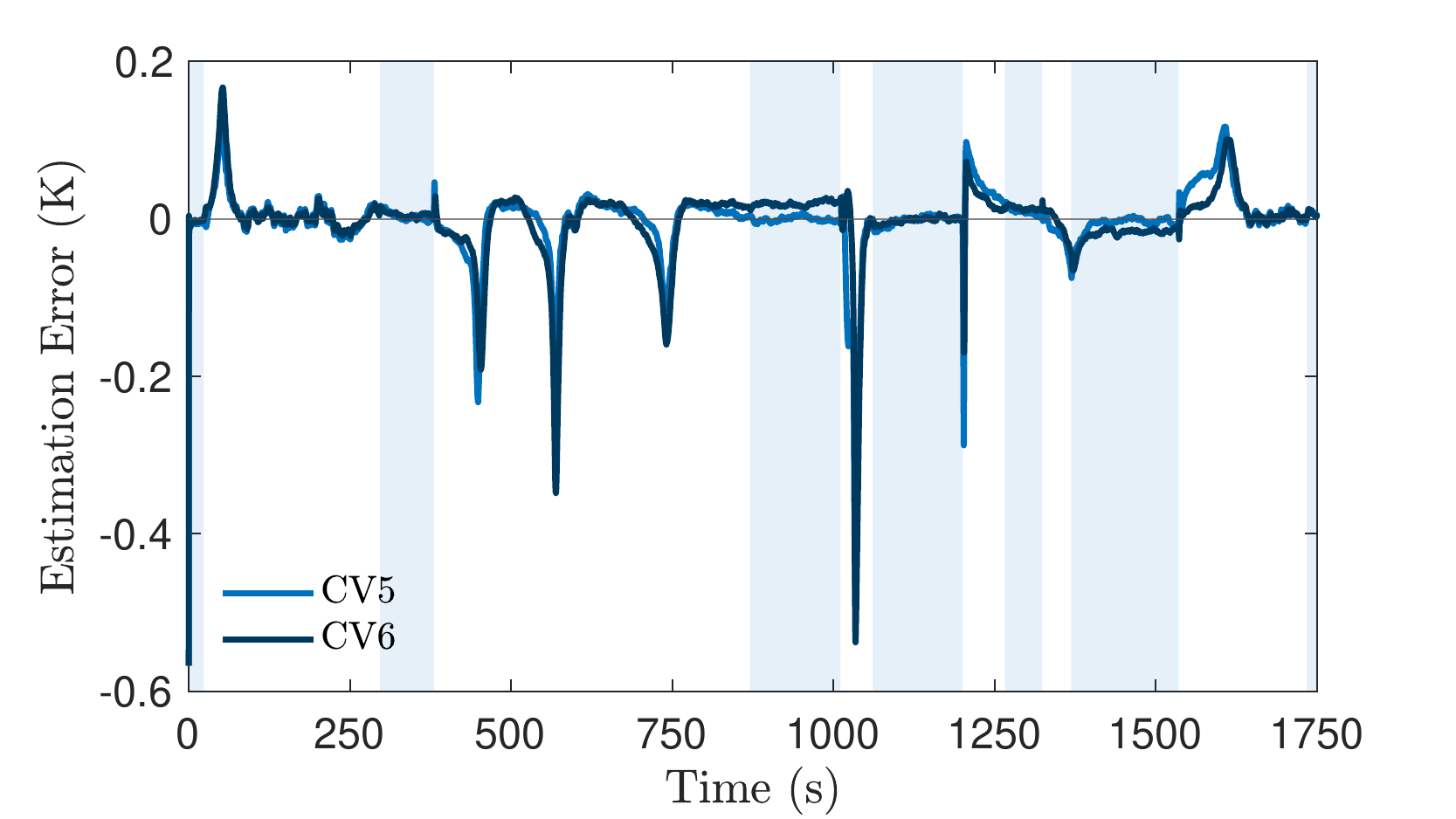}
        \caption{}
        \label{fig:TCall_sim_b}
    \end{subfigure}
    \caption{(a) Estimated temperatures (red) and simulated temperatures (black) at CV5 and CV6. Gray shading denotes the range of temperatures corresponding to the latent region, and the black line is the melting point of the PCM. (b) Estimation errors for CV5 and CV6.}
    \label{fig:TCall_sim}
\end{figure}

Even though the temperatures of CV5 and CV6 are not included in the output vector $y_k$, the SDRE filter is able to track these temperatures with an estimation error magnitude less than 0.6 K. During periods when the temperatures are changing slowly, the estimation error magnitude is less than 0.05 K, which is less than one standard deviation of the measurement noise ($\sqrt{0.0035\textup{ K}^2} = 0.059\textup{ K}$).  In fact, the largest estimation errors only occur during periods of phase change; the larger negative transient estimation errors correspond to periods when the PCM is melting, suggesting that the SDRE filter introduces a small phase lag between the estimates and the true values.  This claim is supported by the observation that the estimation error has a larger negative value when the temperatures change rapidly, such as at $t = 1050$ seconds. In turn, a faster change in temperature results in more phase lag, which increases the transient estimation error.  There are fewer periods of large positive error because the PCM only fully solidifies at the start ($t=50$ seconds) and end ($t=1650$ seconds) of the dataset; this means that the largest transient errors are expected when the experimental TES completes a phase change and begins to experience sensible heating while the state estimates lag behind in the latent temperature range.

\revres{ The primary motivation for using the SDRE as compared to other nonlinear filters is its suitability for the application at hand, namely, SOC estimation for latent TES systems with graph-based models.  Nevertheless, comparison between the SDRE filter, the extended Kalman filter (EKF), and the unscented Kalman filter (UKF) suggests that the SDRE filter achieves the lowest estimation error, as shown in Fig.\ \ref{fig:CV5_comp}, which compares the estimation errors at CV5 for the three state estimators.  Additionally, the two Kalman filters have large estimation errors during the same periods when the SDRE filter has a large estimation error, so we can conclude that these errors result from limitations in the model used for state estimation and not from the choice of state estimator. 
}
\begin{figure}[tb]
    \centering
    \includegraphics[width=3.1in]{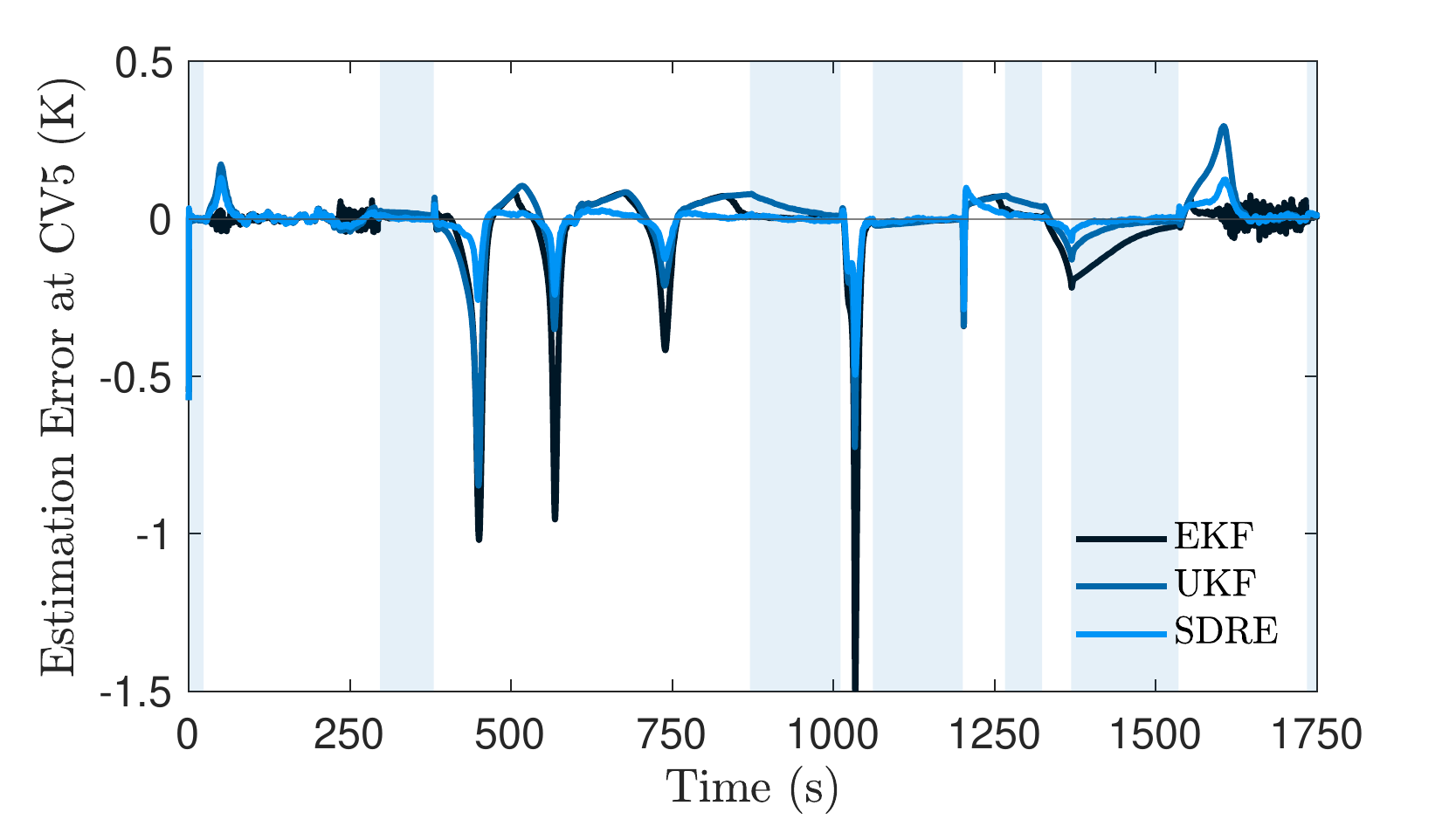}
    \caption{\revres{Comparison of estimation error at CV5 for the EKF, the UKF, and the SDRE filter.}}
    \label{fig:CV5_comp}
\end{figure}

Next, we quantify the temperature estimation error when the simulated measurement from (i) TC1 or (ii) TC3 is withheld from the state estimator. Note that when testing the state estimator on experimental data (see Section \ref{sec:exp_testing}), these thermocouples are withheld to be used for validation.  Fig.\ \ref{fig:TC1_sim_a} shows the temperatures at CV1, calculated by the simulation and estimated by the SDRE filter, when TC1 is withheld, and Fig.\ \ref{fig:TC1_sim_b} shows the corresponding estimation error.  As expected, the estimation error is more substantial when an output is removed; the error magnitude tends to be less than 0.5 K when the control volume temperature is changing slowly, but it increases to nearly 2 K when the temperature changes suddenly, as it does at $t=1050$ seconds or $t=1200$ seconds.  The steady-state estimation errors from $t=850$ to $t=1000$ seconds and $t=1350$ to $t=1550$ seconds are due to the zero mass flow rate through the TES.  Since advection is the dominant form of heat transfer in the fluid channel, when advection is removed, heat transfer in the fluid control volumes is greatly reduced.  This causes the temperature estimate to converge more slowly. 
\begin{figure}[tb]
    \centering
    \begin{subfigure}[b]{\columnwidth}
        \includegraphics[width=3.1in]{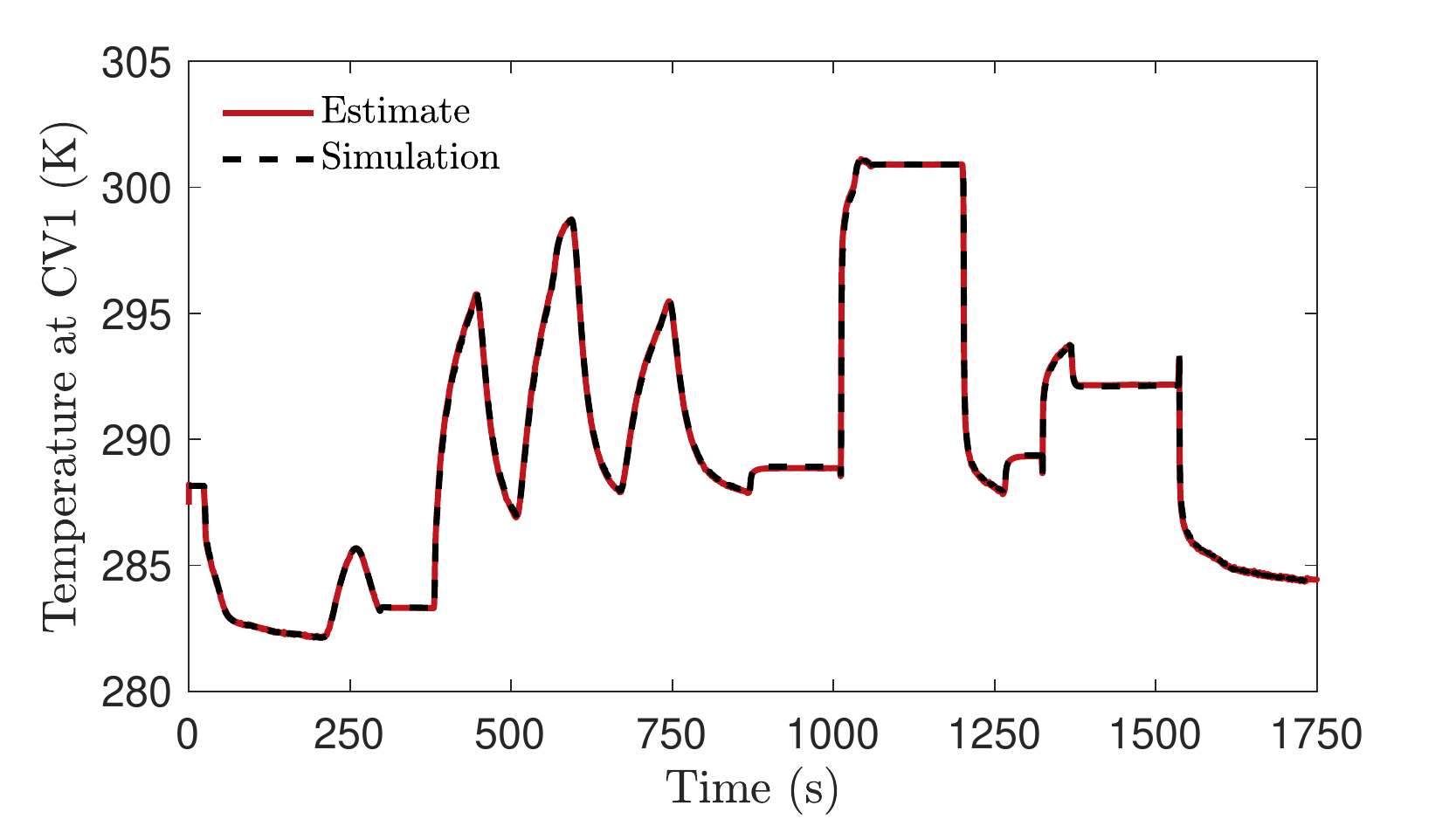}
        \caption{}
        \label{fig:TC1_sim_a}
    \end{subfigure}
    \begin{subfigure}[b]{\columnwidth}
        \includegraphics[width=3.1in]{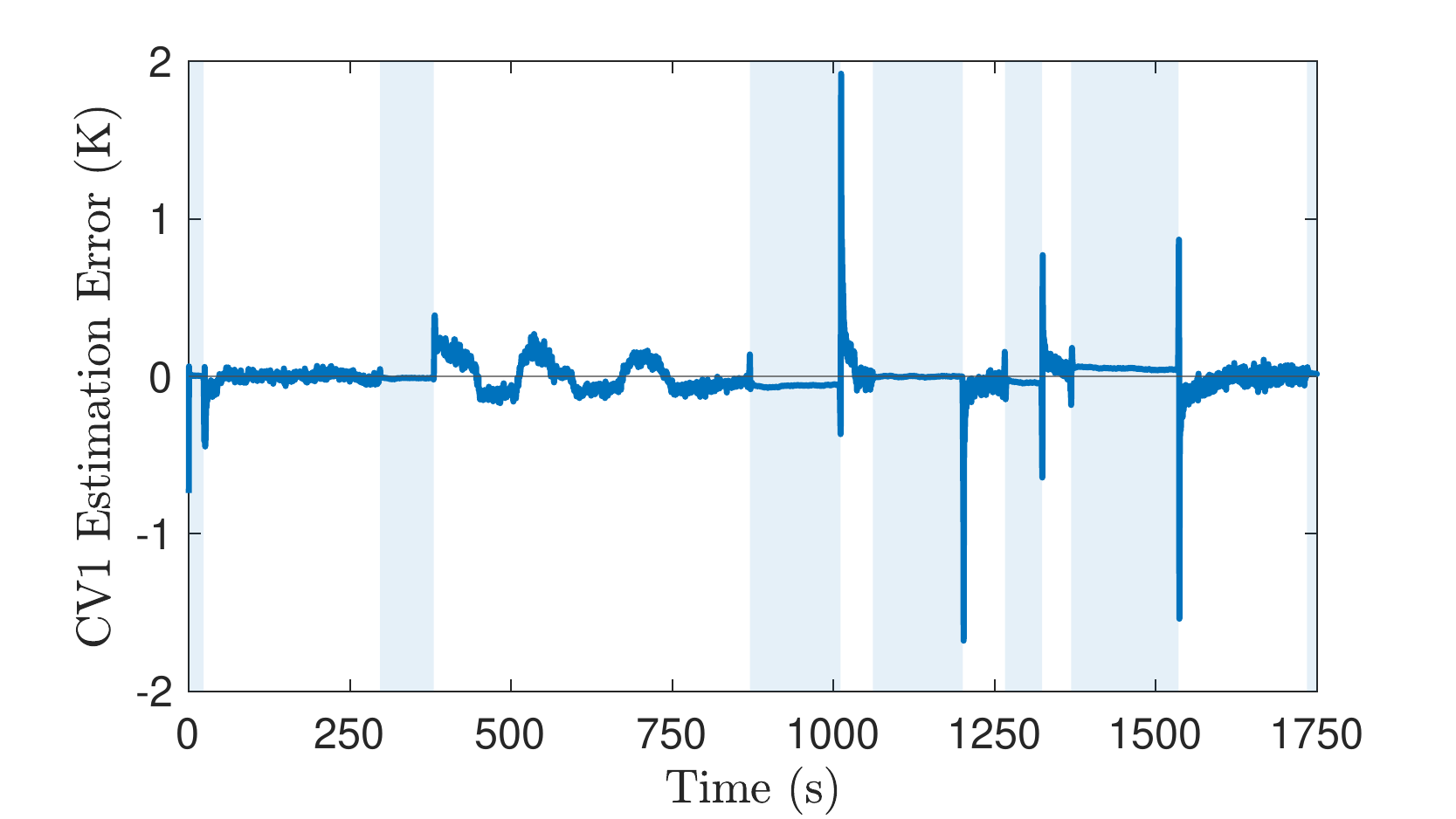}
        \caption{}
        \label{fig:TC1_sim_b}
    \end{subfigure}
    \caption{(a) Simulated temperature at CV1 and the estimated temperature for CV1. (b) Estimation error.}
    \label{fig:TC1_sim}
\end{figure}

Fig.\ \ref{fig:TC3_sim_a} shows that the estimated temperature for control volume CV3 still tracks the simulated temperature for this control volume closely when TC3 is withheld from the estimator; Fig.\ \ref{fig:TC3_sim_b} shows that the error tends to be largest near the melting point and during periods of zero mass flow rate.  The estimation error magnitude stays within 1 K, but it only converges to zero outside the latent temperature range.  Fig.\ \ref{fig:RMSE_sim} shows the root-mean-square error (RMSE) averaged over all $n$ control volumes at time step $t_k$, given in Eqn.\ \eqref{eq:RMSE}.  In this figure, the RMSE is shown for the state estimator with full access to all measurements, the state estimator with access to all measurements except TC1, and the state estimator with access to all measurements except TC3.  The RMS errors of the estimator with all measurements and the estimator without TC1 remain below 0.4 K (except at $t=0$ seconds, when the state estimates are initialized) and are nearly indistinguishable; this indicates that the measurement from TC1 has a negligible effect on the estimates of the other control volumes in the model.  Removing TC3 slightly increases the RMSE during phase changes and periods of zero mass flow rate.  Given that this sensitivity to TC3 is not a result of the system model itself, we verify that it occurs when testing the estimator on the experimental dataset in Section \ref{sec:exp_testing}.
\begin{figure}[tb]
    \centering
    \begin{subfigure}[b]{\columnwidth}
        \includegraphics[width=3.1in]{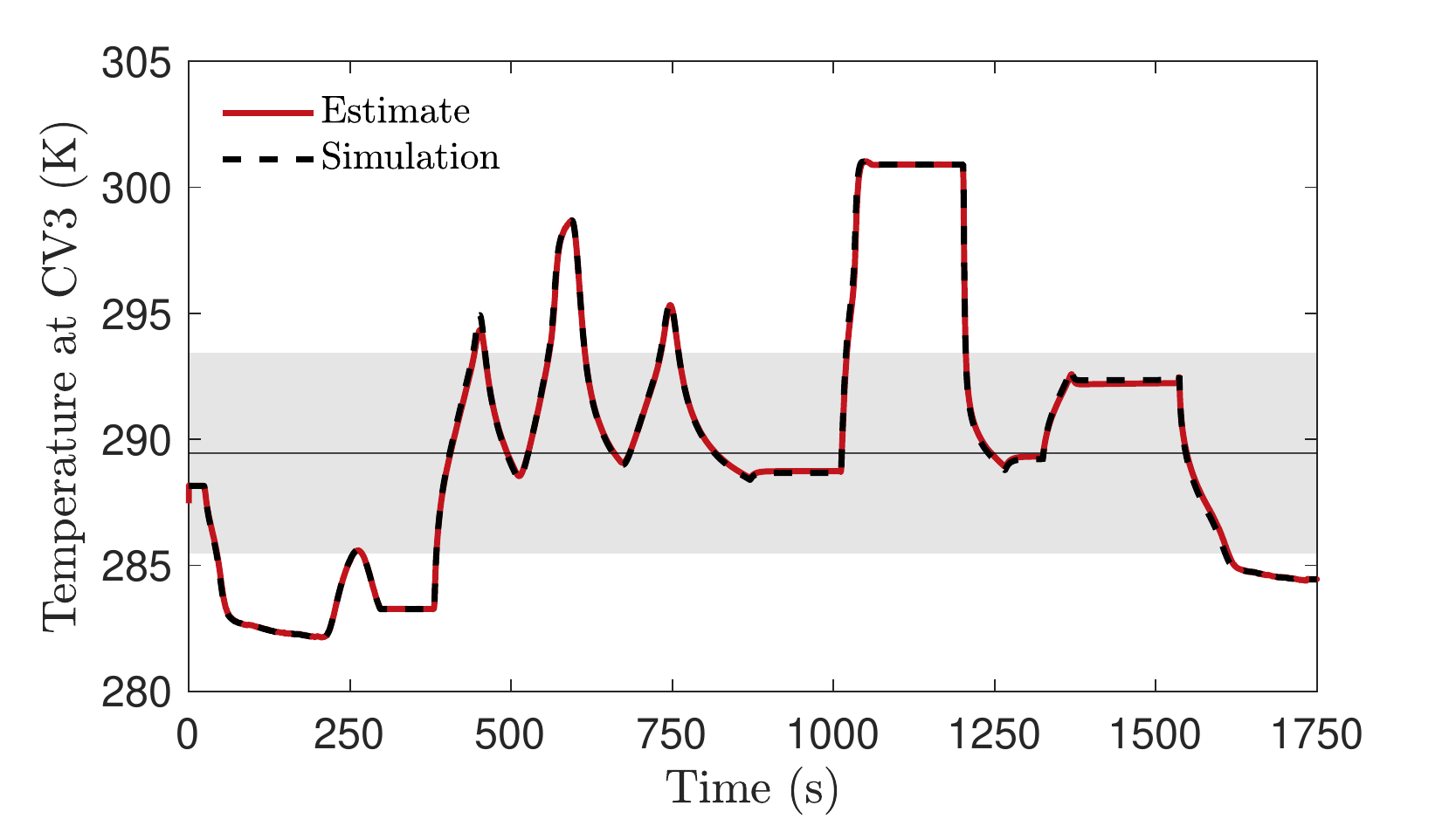}
        \caption{}
        \label{fig:TC3_sim_a}
    \end{subfigure}
    \begin{subfigure}[b]{\columnwidth}
        \includegraphics[width=3.1in]{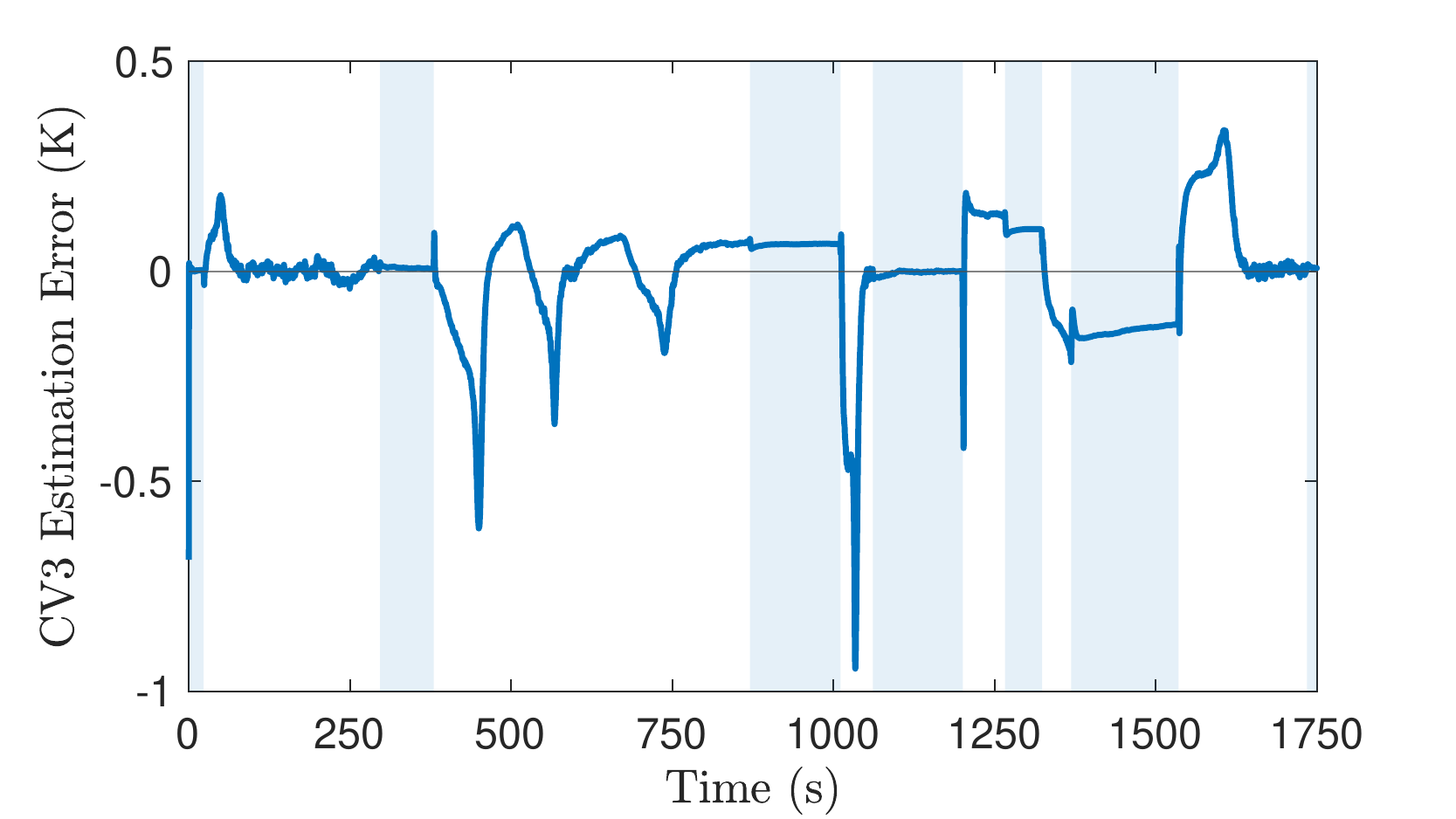}
        \caption{}
        \label{fig:TC3_sim_b}
    \end{subfigure}
    \caption{(a) Simulated temperature at CV3 and the estimated temperature for CV3. (b) Estimation error.}
    \label{fig:TC3_sim}
\end{figure}
\begin{equation}
    e_{rms}(t_k)=\left[\sum_{i=1}^{n}\frac{\left(\hat{x}_i(t_k)-x_i(t_k)\right)^2}{n}\right]^\frac{1}{2}\label{eq:RMSE}
\end{equation}
\begin{figure}[tb]
    \includegraphics[width=3.1in]{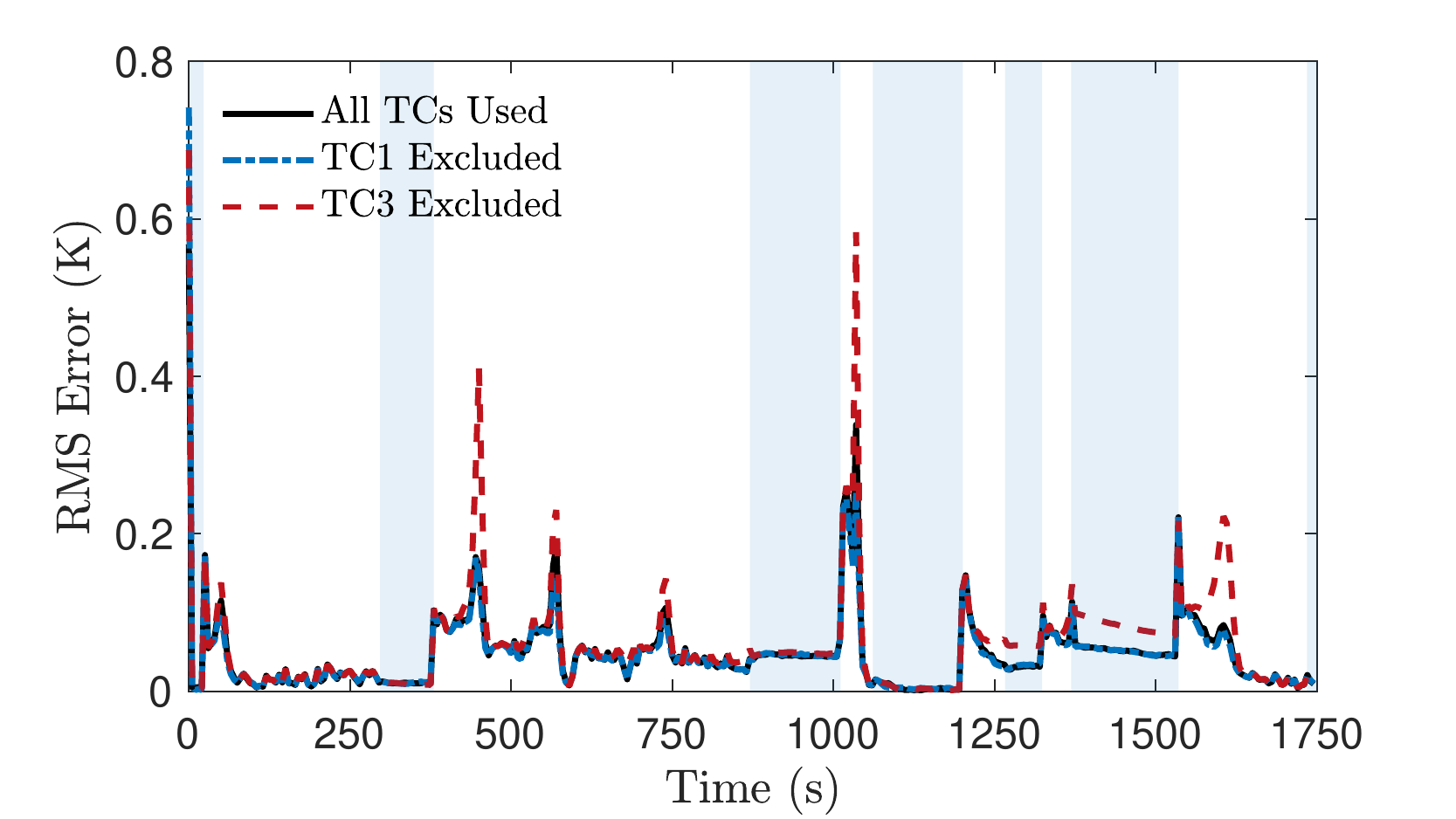}
    \caption{RMS errors of all control volumes for the state estimator with all available measurements, all measurements except TC1, and all measurements except TC3. }
    \label{fig:RMSE_sim}
\end{figure}

Although it is important to analyze the error of individual temperature estimates, the overarching goal of the proposed estimator is to estimate the SOC of the TES for online control decision-making. Fig.\ \ref{fig:SOC_sim_a} compares the simulated and estimated SOC values calculated from the temperature estimates of the SDRE filter with access to all outputs. The shading in Fig.\ \ref{fig:SOC_sim_a} represents the range of SOC in which all or part of the PCM is undergoing a phase change.  Fig.\ \ref{fig:SOC_sim_a} shows simulated PCM temperatures for three control volumes over the duration of the simulation; refer to Fig.\ \ref{fig:TC_CV_labels} for the location of these control volumes.  Comparing Fig.\ \ref{fig:SOC_sim_a} to Fig.\ \ref{fig:OLS_temp} shows that the SOC is near 1 when the PCM temperatures are low and near 0 when the PCM temperatures are high. In Fig.\ \ref{fig:OLS_temp}, the black line is the melting point of the PCM near which small changes in temperature result in large changes in SOC.

This sensitivity near the melting point is also evident in Fig.\ \ref{fig:SOC_sim_b}, which shows the SOC estimation error.  The error is maximized when the PCM temperatures are near the melting point and the SOC is within the latent heat range.  Despite the nonlinear relationship between SOC and temperature, the estimated SOC tracks the simulation well, and the estimate error remains within $\pm$0.02.  Outside the latent range when the system dynamics are approximately linear (between $t=50$ and $t=380$ seconds, for example), the SOC error converges to zero (although there is some noise with amplitude less than 0.001).  The SOC error tends to be positive when the SOC is decreasing and negative when the SOC is increasing, which indicates that there is a small phase lag between the SOC estimate and the true SOC.  Additionally, the error does not converge to zero when the mass flow rate through the TES is zero and the PCM temperatures are within the latent range (for example, between $t=850$ and $t=1000$ seconds).  This is again because heat transfer along the length of the module is dominated by advection. When the mass flow rate is zero, it takes longer for the PCM to reach an equilibrium temperature, which means it takes longer for temperature estimation errors to converge to zero.  Small temperature errors near the phase change temperature result in large errors in SOC.  Outside the latent temperature range, the SOC error converges to zero when the mass flow rate is zero---see the period between $t=1050$ and $t=1200$ seconds for an example.
\begin{figure}[tb]
    \centering
    \begin{subfigure}[b]{\columnwidth}
        \includegraphics[width=3.1in]{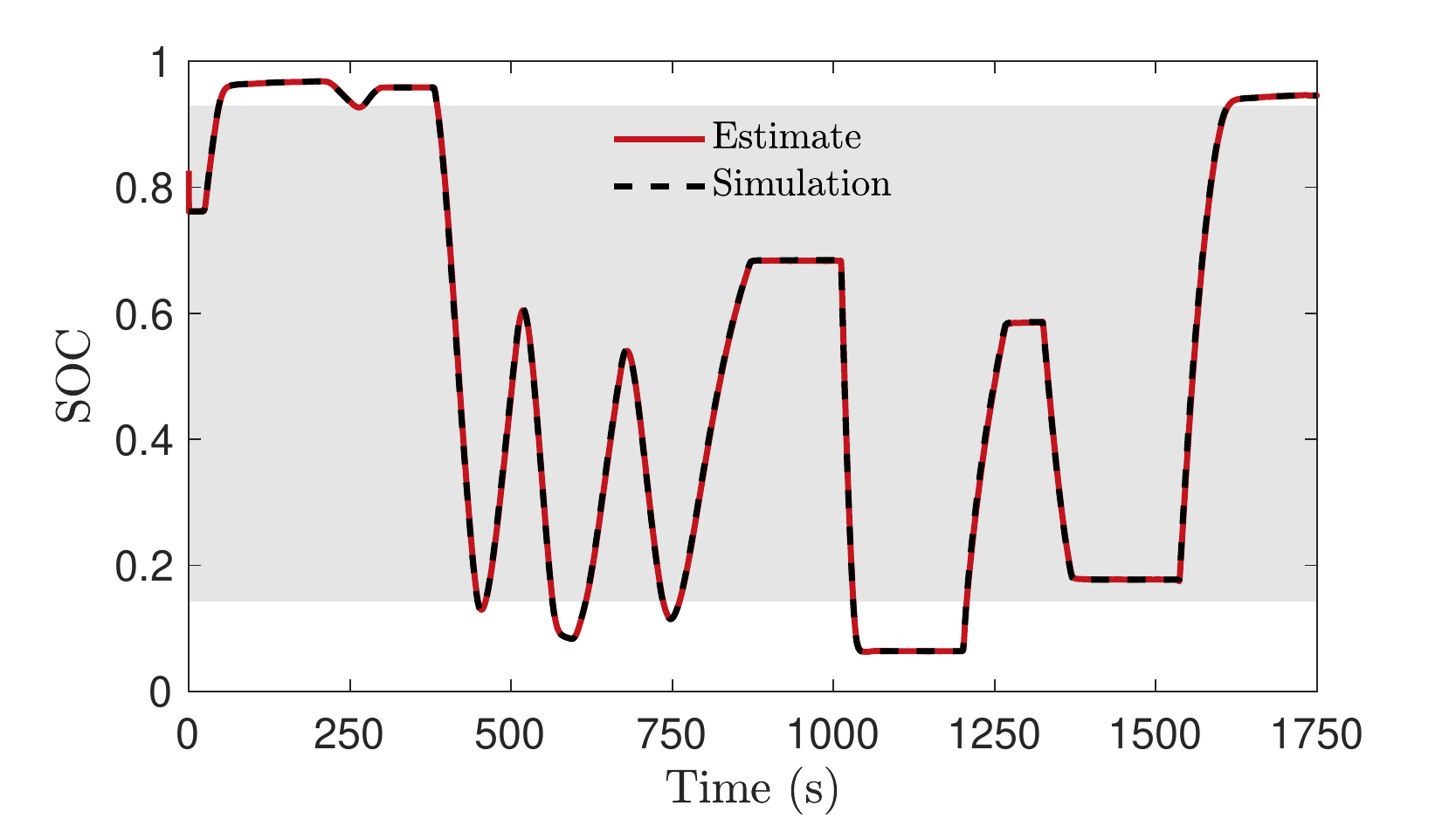}
        \caption{}
        \label{fig:SOC_sim_a}
    \end{subfigure}
    \begin{subfigure}[b]{\columnwidth}
        \includegraphics[width=3.1in]{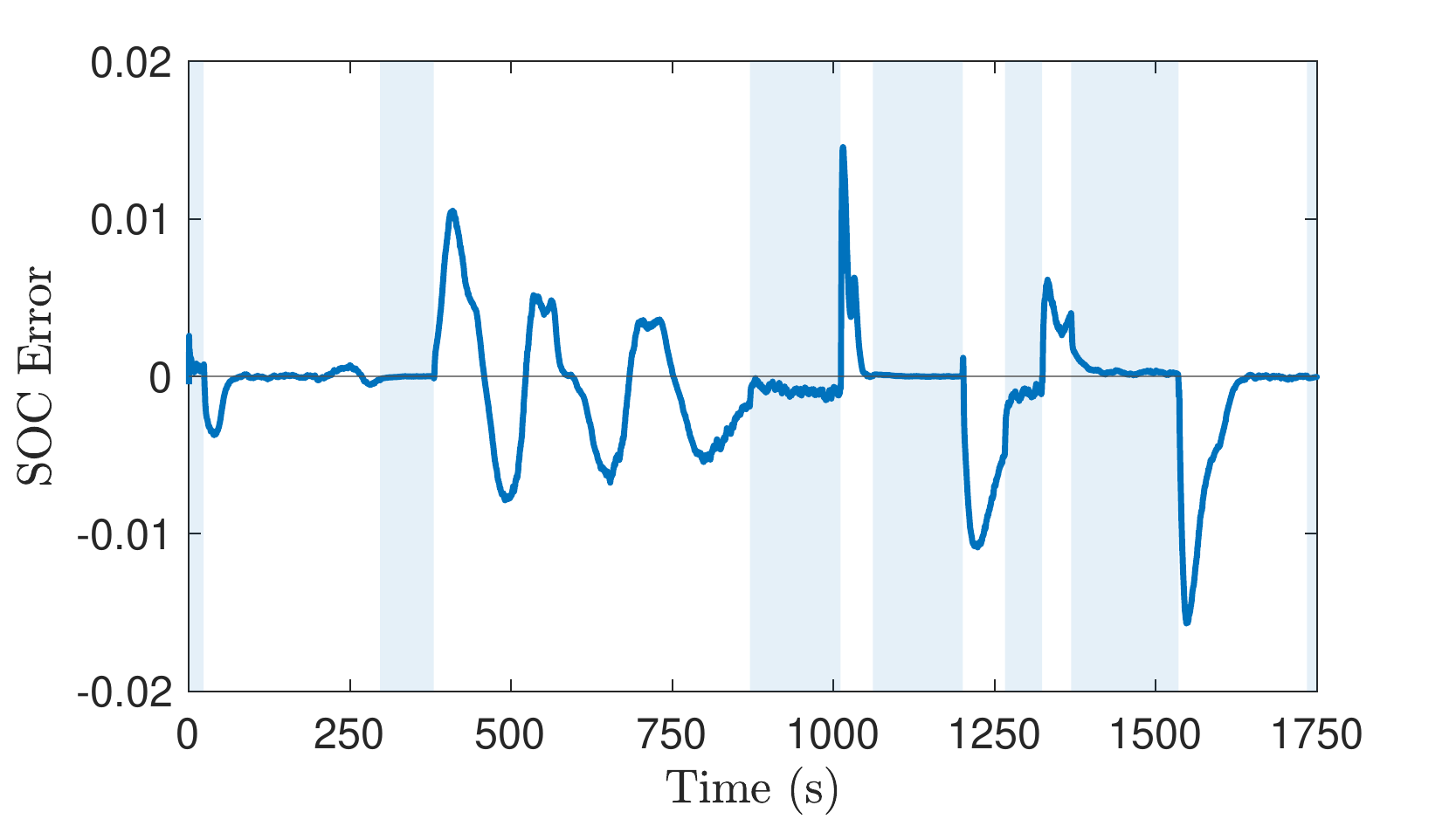}
        \caption{}
        \label{fig:SOC_sim_b}
    \end{subfigure}
    \caption{(a) Simulated and estimated SOC; gray shading represents the latent SOC region where errors are expected to be large. (b) SOC estimation error; blue shading represents periods when the mass flow rate through the TES is zero. }
    \label{fig:SOC_sim}
\end{figure}
\begin{figure}[tb]
    \includegraphics[width=3.1in]{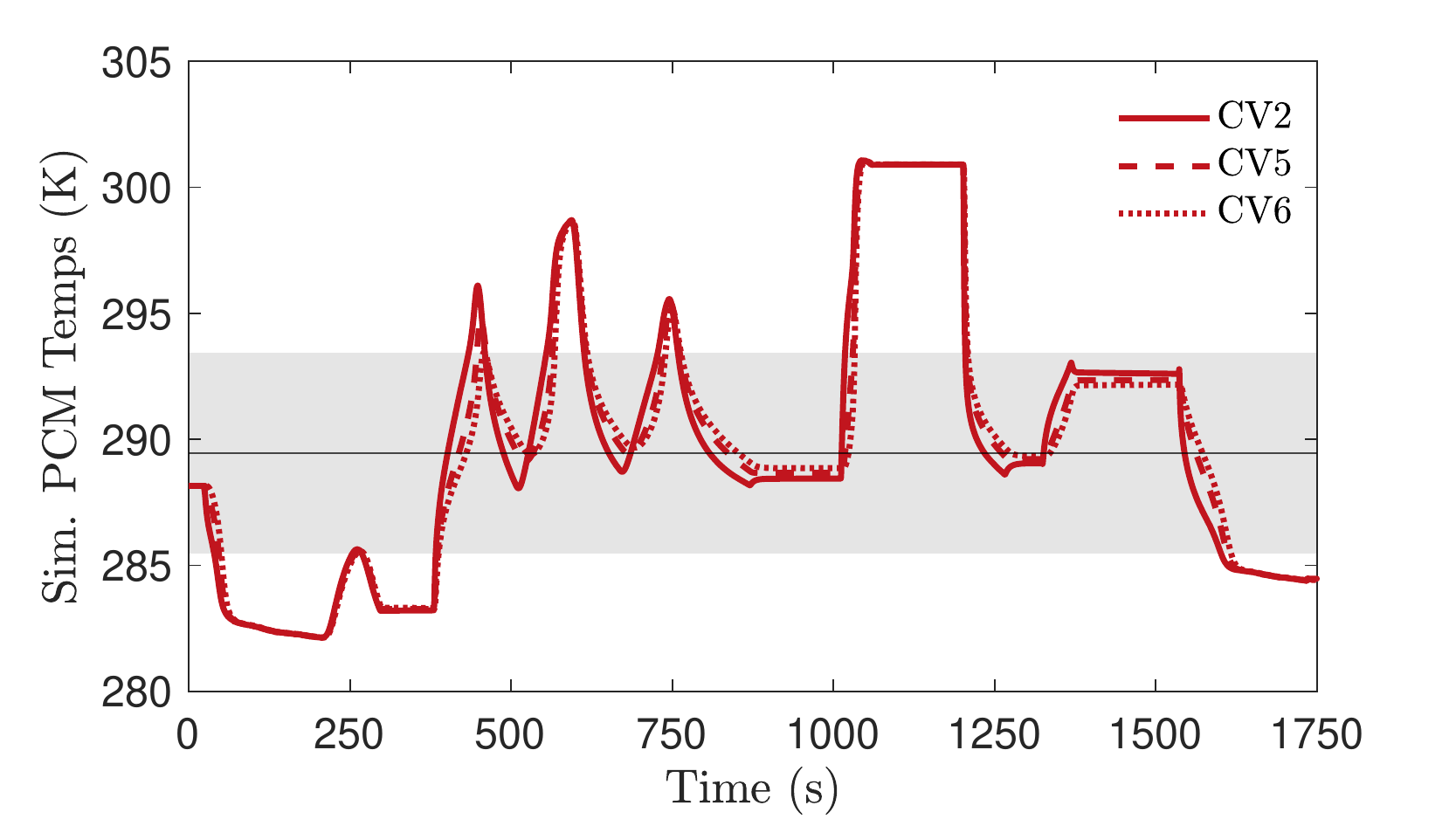}
    \caption{Simulated temperature values from select PCM control volumes.}
    \label{fig:OLS_temp}
\end{figure}


\subsection{Experimental Validation}\label{sec:exp_testing}
In this case study, we test the SDRE filter on experimental data.  First, we investigate convergence of the SDRE filter estimates when (i) thermocouple TC1 is withheld for validation and (ii) thermocouple TC3 is withheld for validation.  We also investigate how the SDRE filter performs with different sample rates.  All estimation errors are calculated as $e(t_k) = \hat{x}(t_k)- y_{k,j}$ where $y_{k,j}$ is a measurement from an excluded thermocouple.

When TC1, the fluid outlet thermocouple measurement, is excluded from the measurement set, the state estimator tracks the fluid outlet temperature closely over long time periods, as shown in Fig.\ \ref{fig:TC1_exp_a}. However, the estimation error increases when the temperature changes rapidly, as shown in Fig.\ \ref{fig:TC1_exp_b} at $t=1000$ seconds, $t=1200$ seconds, and $t=1550$ seconds. For most of the duration of the experiment, the magnitude of the estimation error at CV1 is less than 1 K.  This error behavior is similar to that of the simulated results shown in Fig.\ \ref{fig:TC1_sim_b}, but the larger magnitude of the error when implementing the estimator on the experimental testbed is due to modeling errors. These include inaccuracies in the reduced-order finite volume model due to spatial discretization \cite{gohil_reduced-order_2020, shanks_design_2022}, simplifications of the phase-change dynamics such as neglecting undercooling and hysteresis \cite{sgreva_thermo-physical_2022, barz_paraffins_2021}, and exogenous disturbances in the physical system like heat transfer with the surroundings.
\begin{figure}[tb]
    \centering
    \begin{subfigure}[b]{\columnwidth}
        \includegraphics[width=3.1in]{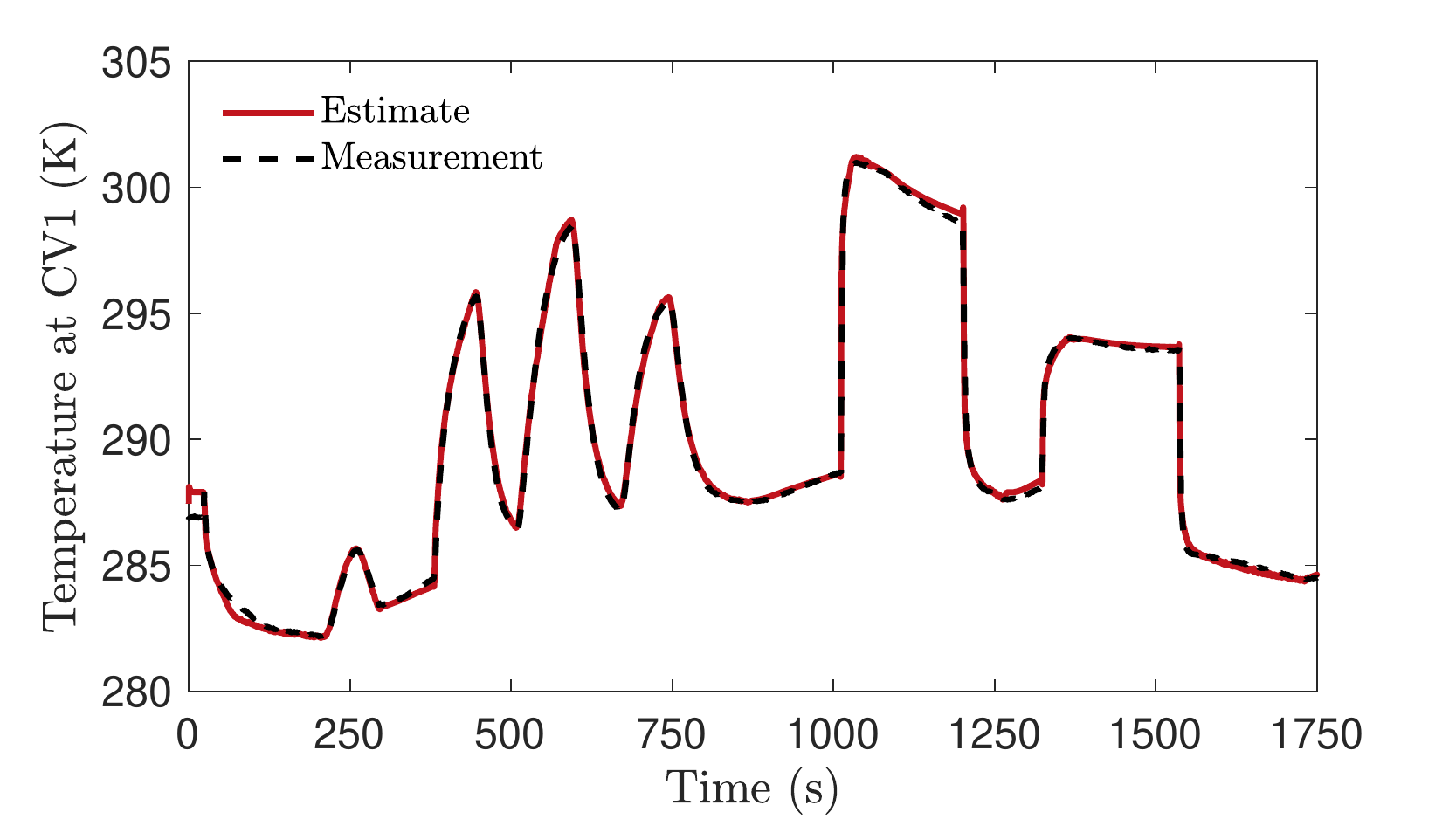}
        \caption{}
        \label{fig:TC1_exp_a}
    \end{subfigure}
    \begin{subfigure}[b]{\columnwidth}
        \includegraphics[width=3.1in]{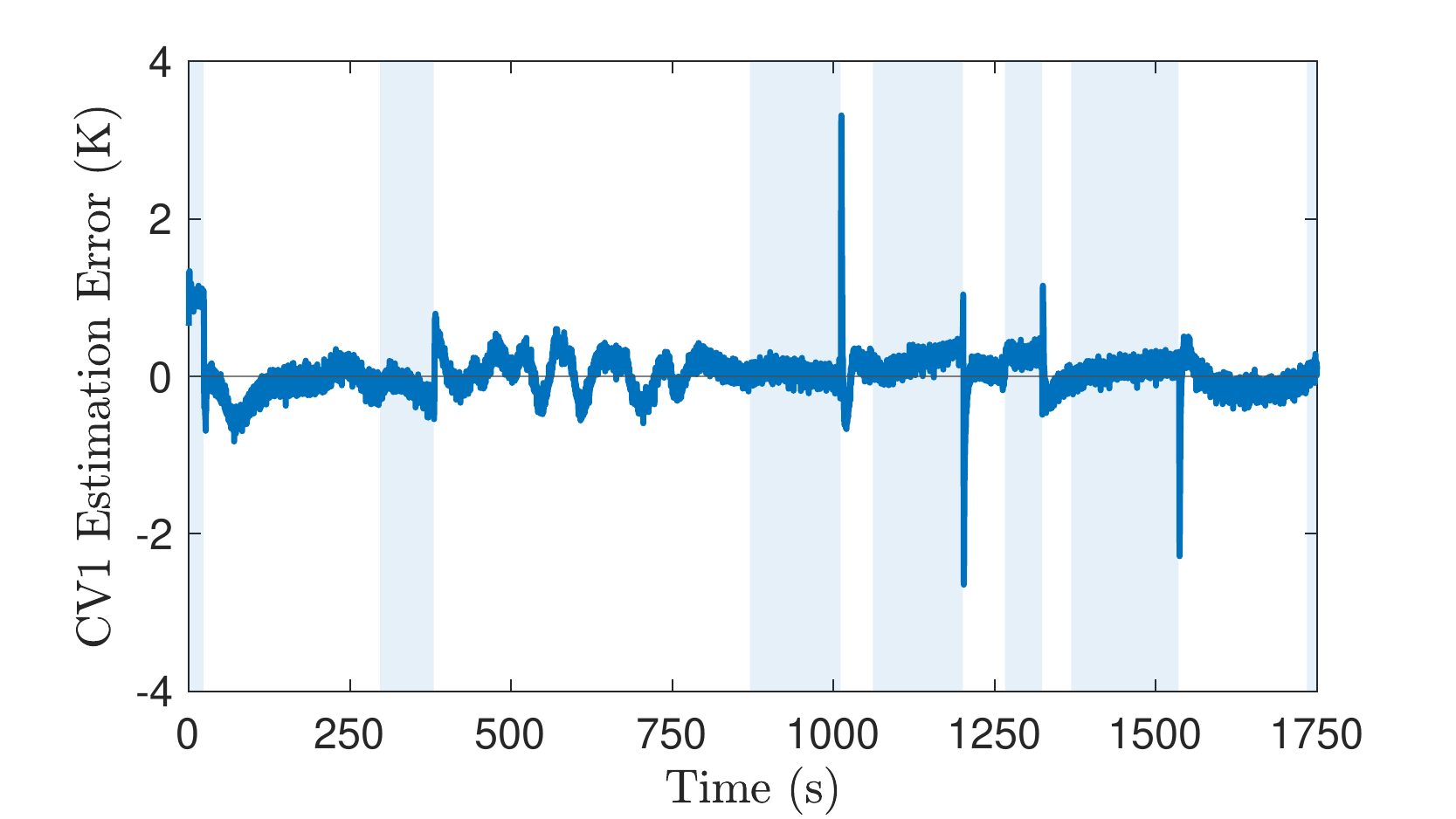}
        \caption{}
        \label{fig:TC1_exp_b}
    \end{subfigure}
    \caption{(a) Estimated temperature of CV1 compared to the fluid outlet temperature measurement TC1. (b) Estimate error for CV1.}
    \label{fig:TC1_exp}
\end{figure}

Excluding TC3 from the measurement set results in degradation of the estimator's performance, as expected based on the simulated results. The estimate of CV3 still tracks TC3 closely over long time periods, as Fig.\ \ref{fig:TC3_exp_a} shows.  When the TES is experiencing rapid changes in temperature, Fig.\ \ref{fig:TC3_exp_b} shows that the magnitude of the error tends to increase; the convergence rate of the estimate error of CV3 is slower than the transient dynamics, resulting in larger errors up to 6 K at $t=1050$ seconds and 3 K at $t=1550$ seconds.  During the period between $t=450$ and $t=900$ seconds and the period between $t=1200$ and $t=1400$ seconds, the PCM undergoes partial phase change.  Estimation errors during these periods of time are due to the phase change hysteresis and undercooling effects which delay solidification when the PCM is cooled rapidly \cite{barz_paraffins_2021}.  The system model does not account for these effects, so the estimator incorrectly predicts that CV3 (and other unmeasured control volumes) begins to solidify at 293.5 K (the top boundary of the gray shaded region in Fig.\ \ref{fig:TC3_exp_a}) whereas the PCM in the experimental TES begins to solidify at a lower temperature.  Over sufficiently long time periods during which the PCM temperatures reach equilibrium, such as $t=750$ to $t=1000$ seconds, errors caused by hysteresis and undercooling are corrected and the estimate error returns to zero.
\begin{figure}
    \centering
    \begin{subfigure}[b]{\columnwidth}
        \includegraphics[width=3.1in]{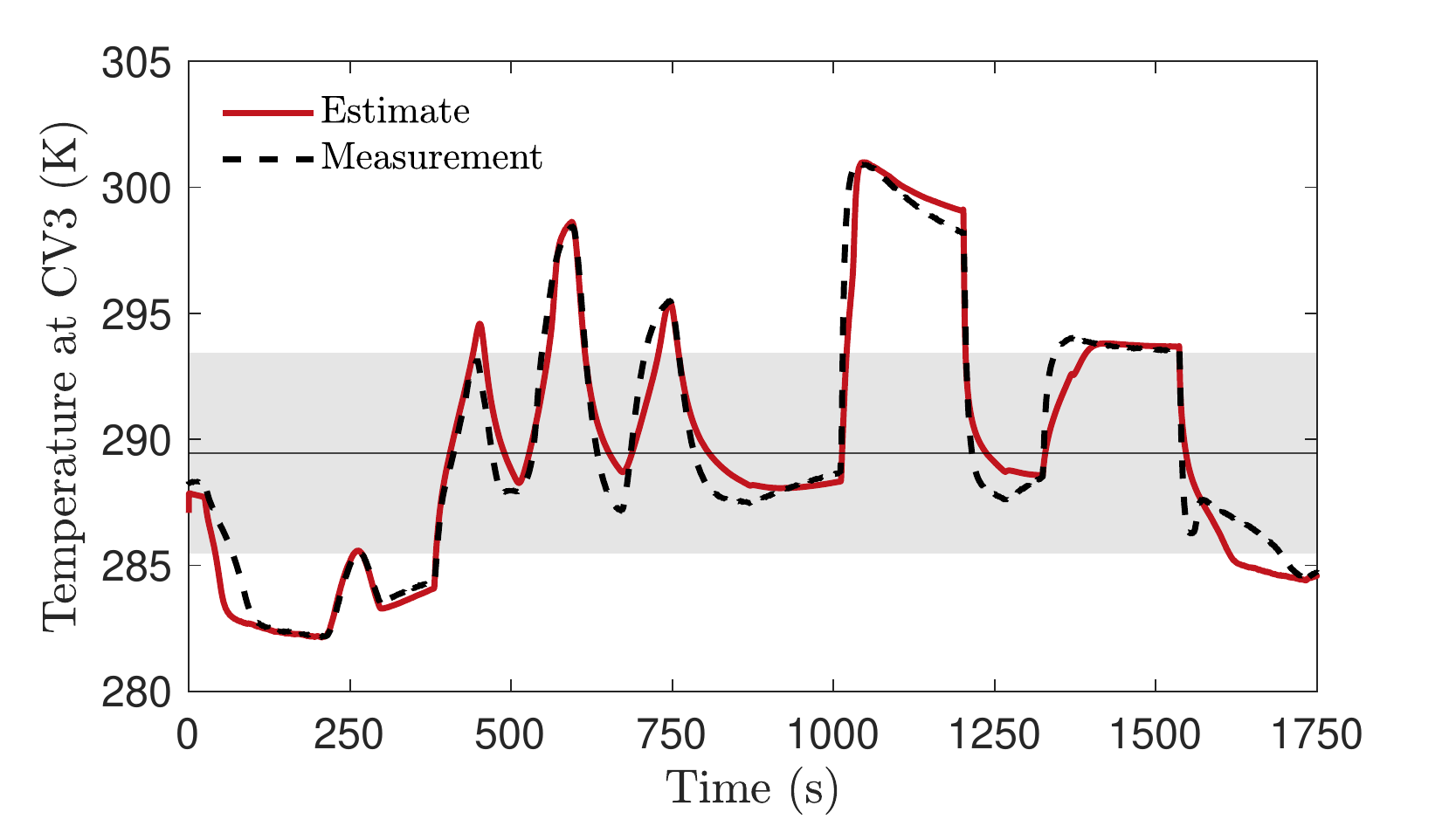}
        \caption{}
        \label{fig:TC3_exp_a}
    \end{subfigure}
    \begin{subfigure}[b]{\columnwidth}
        \includegraphics[width=3.1in]{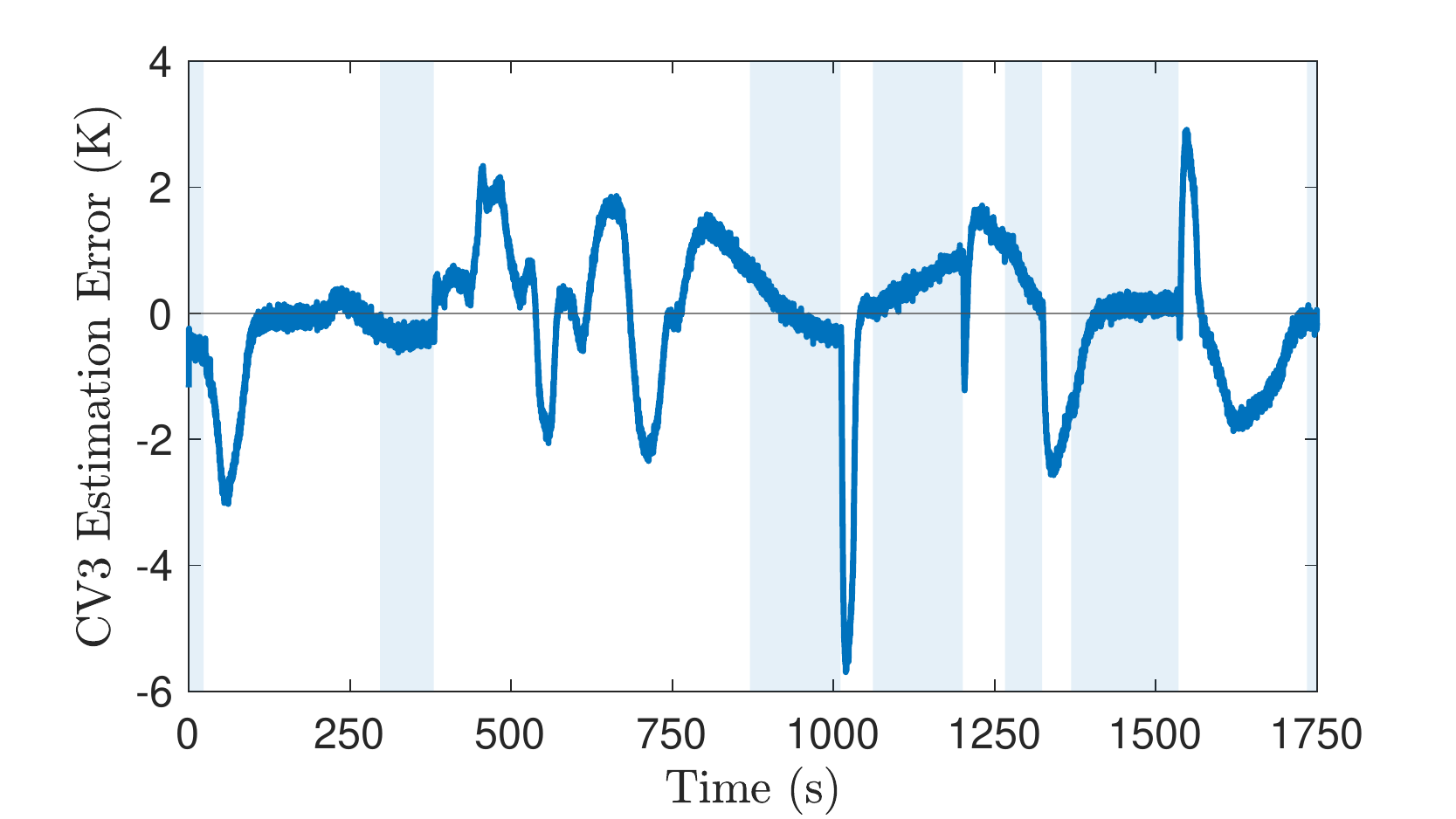}
        \caption{}
        \label{fig:TC3_exp_b}
    \end{subfigure}
    \caption{(a) Estimated temperature of CV3 compared to the measurement of TC3. (b) Estimate error for CV3.}
    \label{fig:TC3_exp}
\end{figure}

\subsection{Effect of Sample Rate}\label{sec:sample_rate}
In this section, we investigate how the SDRE filter performs when measurements are provided at four different rates: 80 samples per second (S/s), 10 S/s, 1 S/s, and 0.2 S/s.  All previous experimental and simulation results used a sample rate of 10 S/s.  The continuous-discrete SDRE filter formulation allows for more flexibility with respect to sample rates because the state estimates and the estimate covariance matrix are propagated with the state transition matrix of the continuous-time system; estimation errors introduced by discretizing the nonlinear system are mitigated by propagating the state estimates and error covariance with a time step smaller than the interval between samples. 

Fig.\ \ref{fig:TC1_exp_SR} shows (a) the estimated temperatures at CV1 and (b) the estimation errors when TC1 is excluded from the measurement set during a validation test with the experimental dataset.  As discussed in the previous section, the state estimator can track the temperature at CV1 well; this is also the case at the faster and slower sample rates.  In Fig.\ \ref{fig:TC1_exp_SR_b}, 80 S/s, 10 S/s, and 1 S/s all appear to have similar estimation errors, but 0.2 S/s produces larger error during periods of transient temperature changes.  To more easily compare the performance of the four sample rates, we calculate the root-mean-square errors for the temperature at CV1 averaged over all time steps using Eqn.\ \eqref{eq:RMSE_t} where $N_s$ is the total number of time steps in the dataset for the given sample rate. Results for the four sample rates are included in Table \ref{tab:RMSE_TC1}.  Note that this RMSE definition represents an average over all time steps, but the previous RMSE definition in Eqn.\ \eqref{eq:RMSE} averages over all control volumes.  As expected, the 80 S/s rate results in the smallest estimation error, but only by a small margin; 10 S/s results in a similar RMSE.  The slowest sample rate, 0.2~S/s, results in almost twice the RMSE of the next slowest sample rate, 1 S/s.  In the fluid channel, the control volume temperatures can change rapidly due to the high advective heat transfer rates.  The larger error in the 0.2~S/s sample rate is due to the fast dynamics of control volume CV1; the estimation error can increase quickly when update steps are not performed fast enough.
\begin{equation}\label{eq:RMSE_t}
    e_{rms,j}=\left[\sum_{k=1}^{N_{s}}\frac{\left(\hat{x}_{j}(t_k)-y_{k,j}\right)^2}{N_{s}}\right]^\frac{1}{2}
\end{equation}
\begin{table}[htbp]
\centering
\caption{RMS errors at CV1 for four sample rates}
\label{tab:RMSE_TC1}
\begin{tabular}{r|l}
Sample Rate  & $e_{rms,CV1}$ (K) \\ \hline
80 S/s         & 0.2436     \\
10 S/s         & 0.2617     \\
1 S/s          & 0.3811    \\
0.2 S/s         & 0.7526    
\end{tabular}
\end{table}

Fig.\ \ref{fig:SOC_exp_TC1_SR} shows the estimated SOC for the four sample rates.  All four estimates are similar with almost no difference outside the latent range. During the periods of partial phase change, a noticeable difference in the peak of the SOC estimate can be seen at $t=500$ seconds, $t=650$ seconds, and $t=1250$ seconds. The faster sample rates reach a sharper and higher peak SOC, suggesting that the use of a faster sample rate results in a faster response time to transient changes.
\begin{figure}[tb]
    \centering
    \begin{subfigure}[b]{\columnwidth}
        \includegraphics[width=3.1in]{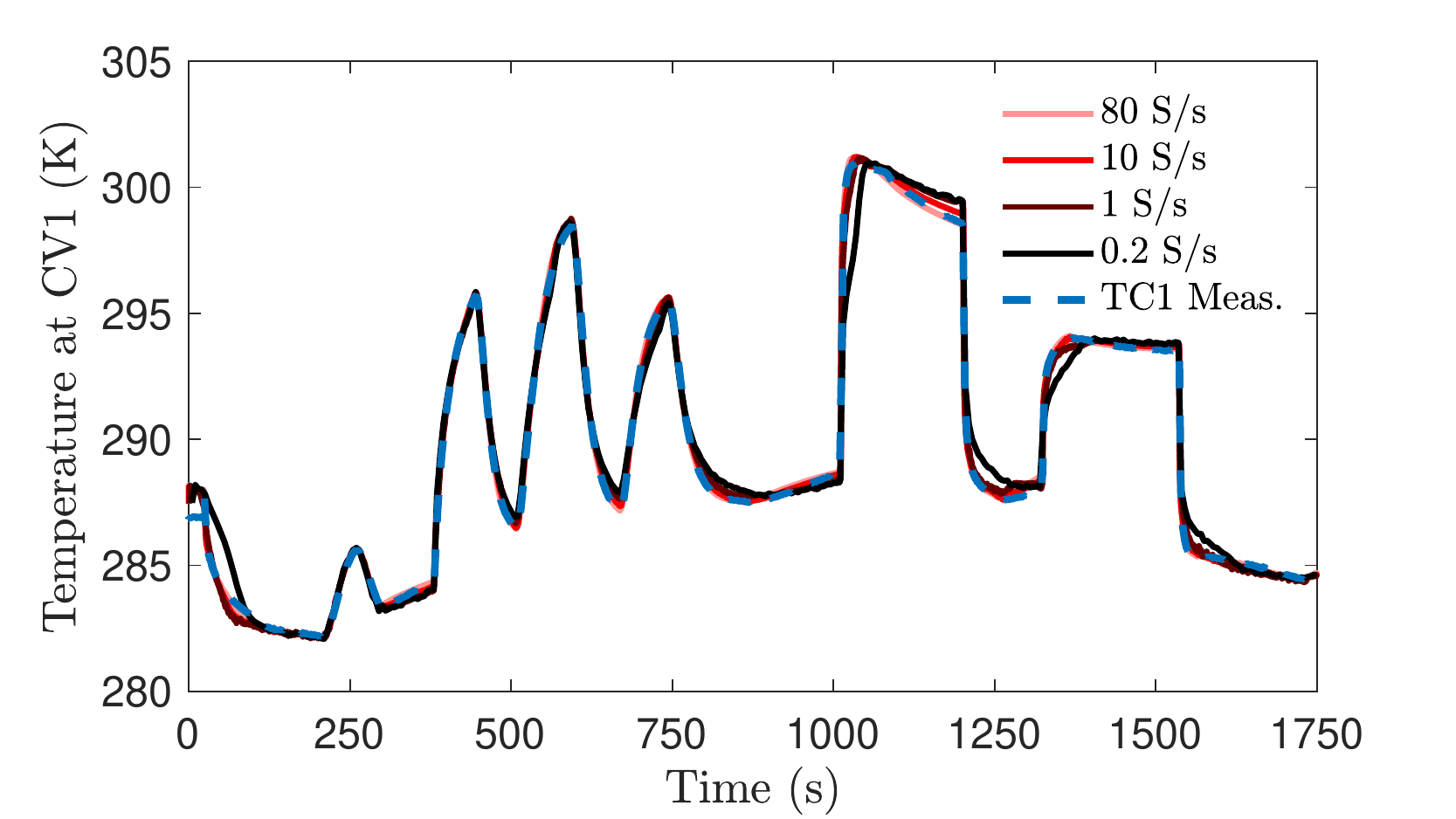}
        \caption{}
        \label{fig:TC1_exp_SR_a}
    \end{subfigure}
    \begin{subfigure}[b]{\columnwidth}
        \includegraphics[width=3.1in]{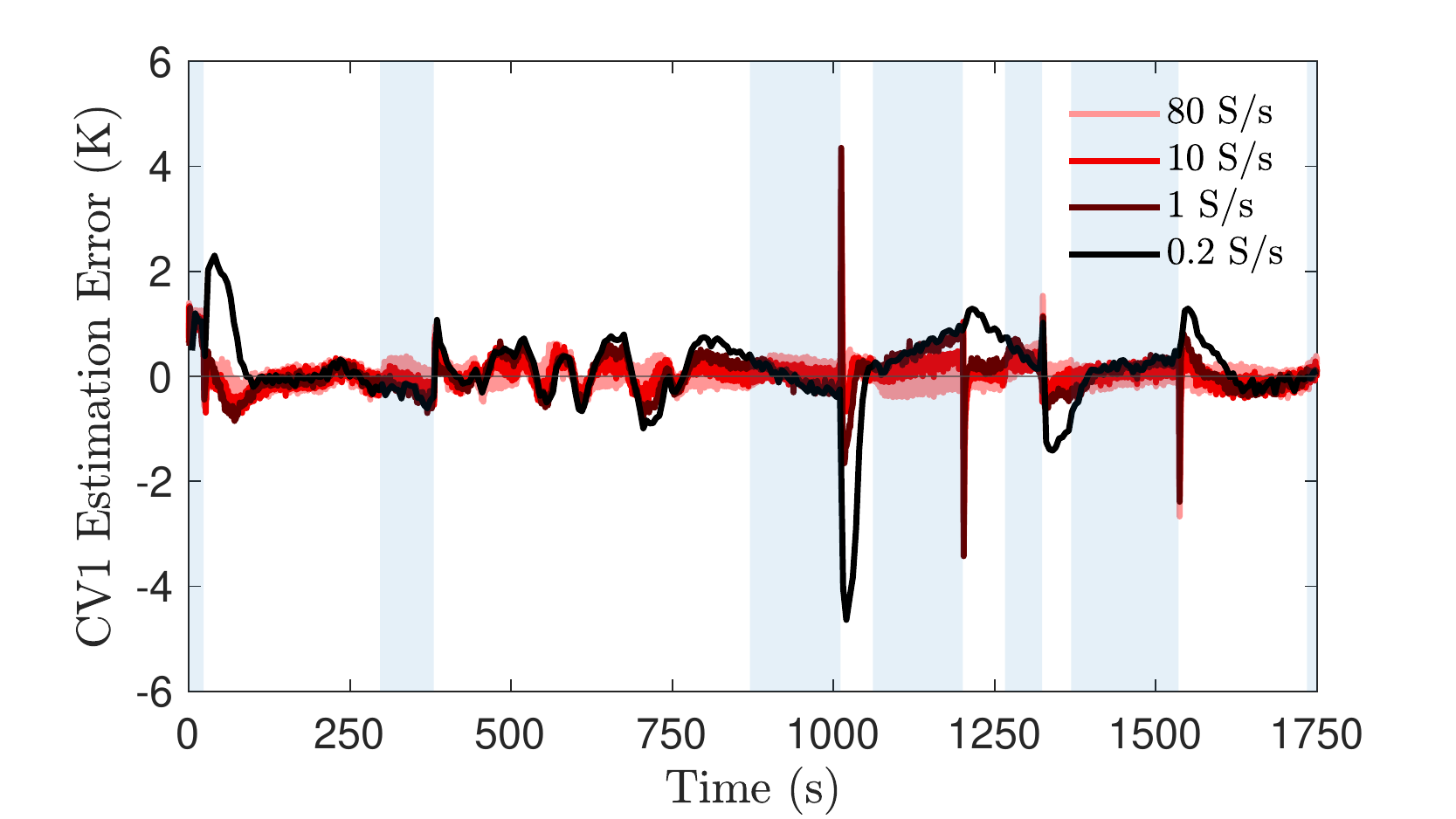}
        \caption{}
        \label{fig:TC1_exp_SR_b}
    \end{subfigure}
    \caption{(a) Estimated temperatures of CV1 compared to the fluid outlet temperature measurement TC1. (b) Estimation errors for CV1.}
    \label{fig:TC1_exp_SR}
\end{figure}
\begin{figure}[tb]
    \includegraphics[width=3.1in]{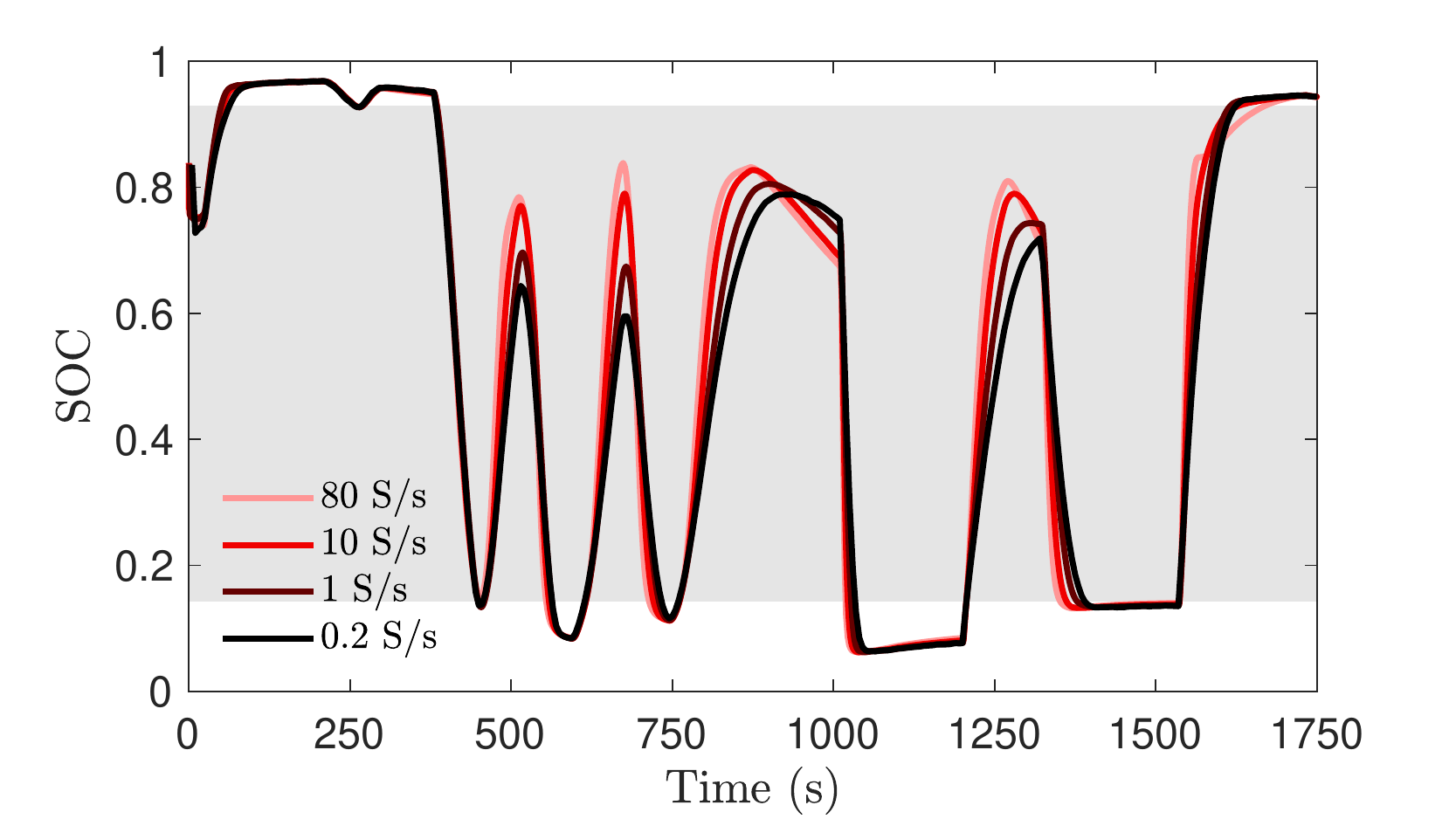}
    \caption{Estimated SOC for the four sample rates when TC1 is withheld.}
    \label{fig:SOC_exp_TC1_SR}
\end{figure}

Finally, we compare the performance of the four sample rates when TC3 is excluded from the measurement set.  Fig.\ \ref{fig:TC3_exp_SR} shows (a) the estimated temperatures at CV3 and (b) the estimation errors.  In Fig.\ \ref{fig:TC3_exp_SR_b}, it is clear that the state estimator receiving 80 S/s has a slightly smaller error most of the time.  This is noticeable at $t=500$ seconds, $t=1100$ seconds, and $t=1600$ seconds.  The RMSE of the estimate at CV3 averaged over all time steps is given in Table \ref{tab:RMSE_TC3}. The state estimator with a sample rate of 80 S/s again performs the best, but by a small margin; this time, there is only a small difference in RMSE between the fastest and slowest sample rates.  Since the thermocouple TC3 is embedded in the PCM layer, the dynamics of control volume CV3 are slower than the dynamics of CV1, which explains why sample rate has a lesser effect on RMSE at CV3 compared to the RMSE at CV1 in Table \ref{tab:RMSE_TC1}.  Note that the RMS errors in Table \ref{tab:RMSE_TC3} are larger than those in Table \ref{tab:RMSE_TC1}; this indicates that the measurement from TC3 is more critical to estimate the SOC accurately than that of TC1.  This makes sense given that TC3 is located near the center of the TES; without the measurement from TC3, only measurements of the temperatures at the periphery of the PCM layer are available. Estimation errors in CV3 and the control volumes above CV3 can be large for long time periods (greater than 2 K for more than 20 seconds, according to Fig.\ \ref{fig:TC3_exp_SR_b}) before the state estimator compensates for the errors.
\begin{figure}[tb]
    \centering
    \begin{subfigure}[b]{\columnwidth}
        \includegraphics[width=3.1in]{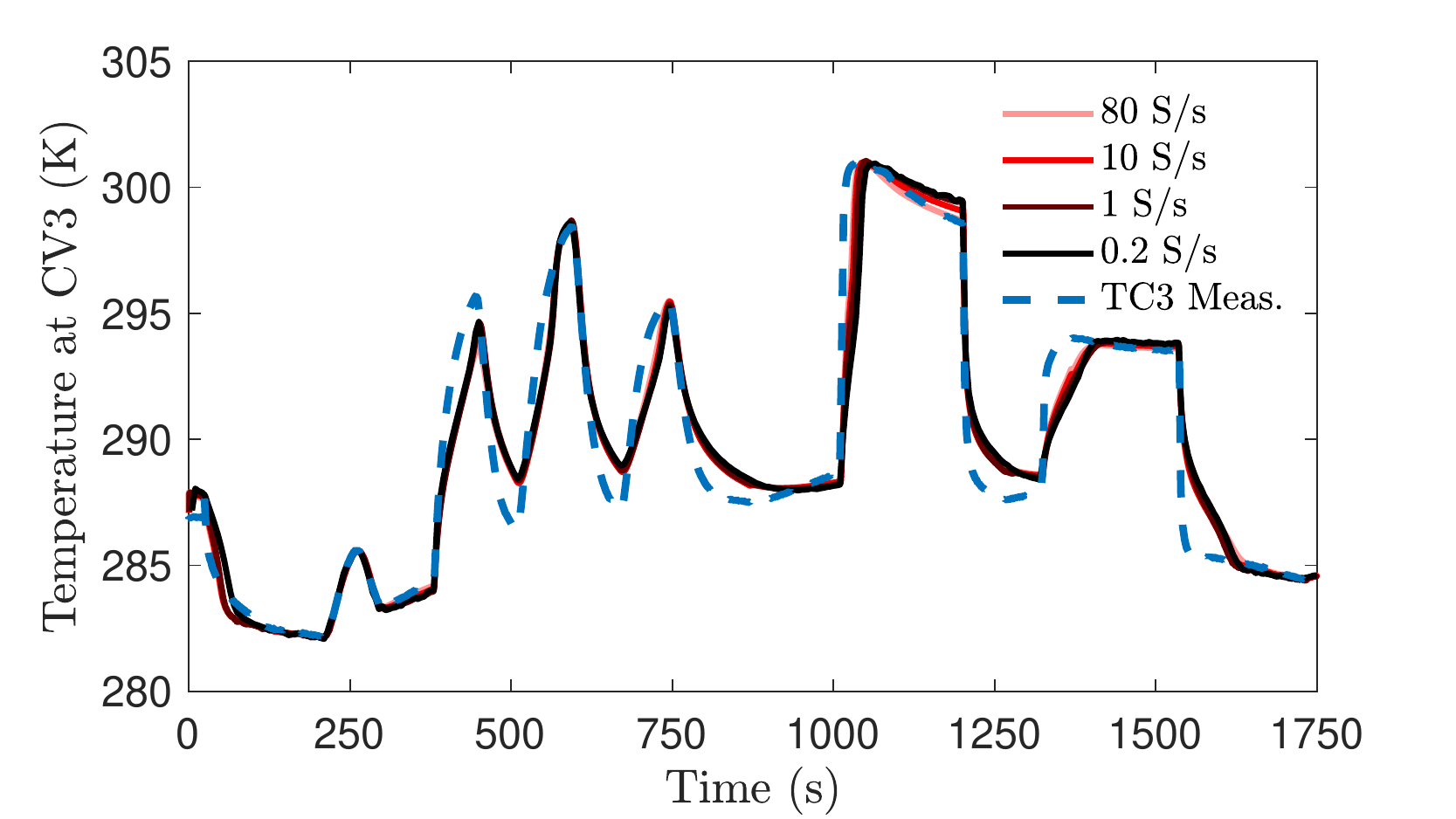}
        \caption{}
        \label{fig:TC3_exp_SR_a}
    \end{subfigure}
    \begin{subfigure}[b]{\columnwidth}
        \includegraphics[width=3.1in]{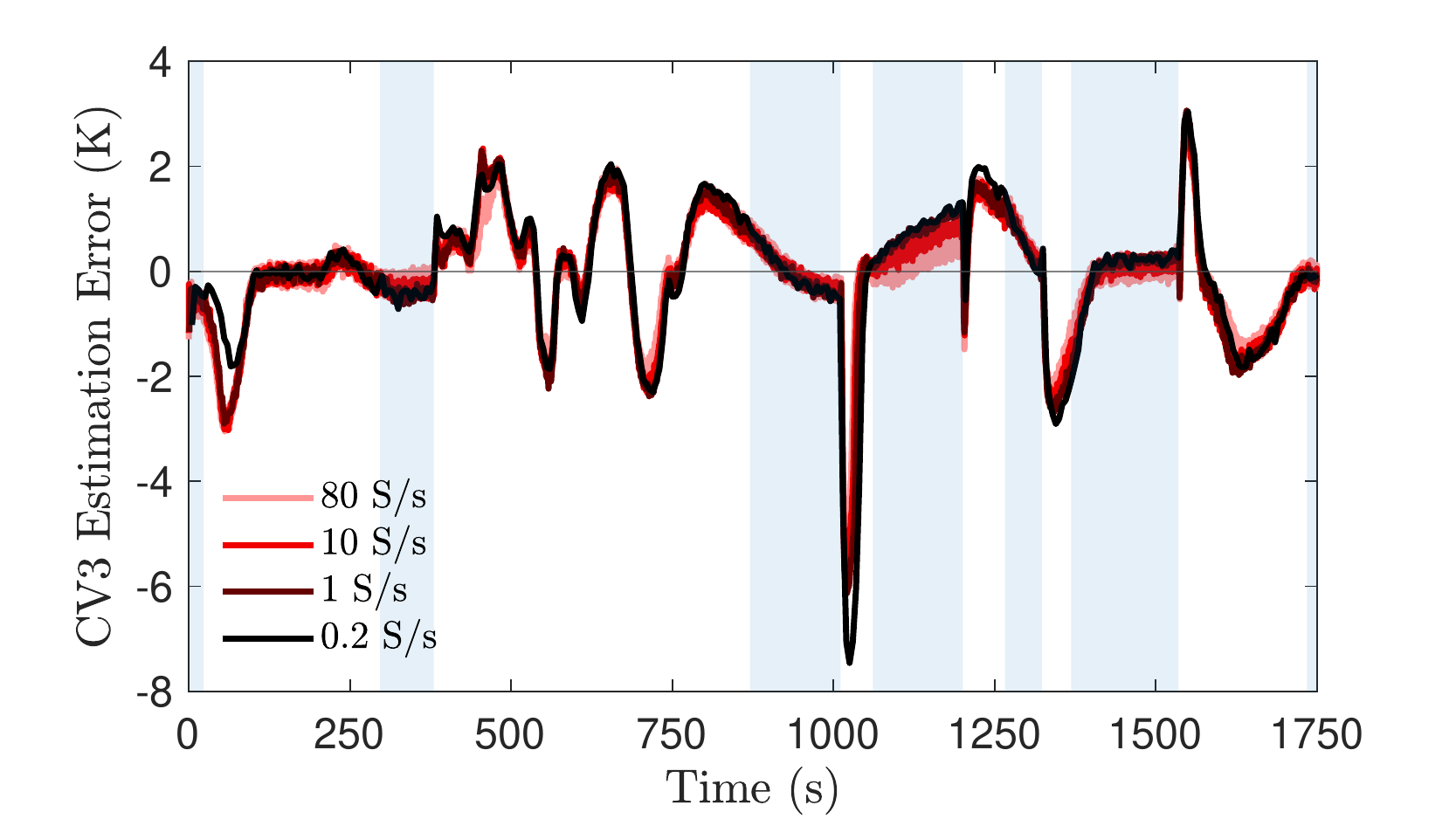}
        \caption{}
        \label{fig:TC3_exp_SR_b}
    \end{subfigure}
    \caption{(a) Estimated temperature of CV3 compared to the measurement of TC3. (b) Estimation errors for CV3.}
    \label{fig:TC3_exp_SR}
\end{figure}
\begin{table}[htbp]
\centering
\caption{RMS errors at CV3 for four sample rates}
\label{tab:RMSE_TC3}
\begin{tabular}{r|l}
Sample Rate & $e_{rms,CV3}$ (K) \\ \hline
80 S/s         & 1.0022     \\
10 S/s         & 1.1046     \\
1 S/s          & 1.2048     \\
0.2 S/s         & 1.3270
\end{tabular}
\end{table}

Fig.\ \ref{fig:SOC_exp_TC3_SR} shows the estimated SOC for the four sample rates when the measurement from TC3 is excluded.  The results shown in this figure are almost identical to those in Fig.\ \ref{fig:SOC_exp_TC1_SR}, except there is less difference between the 80 S/s results and the 10 S/s results.  These results demonstrate an important advantage of the continuous-discrete SDRE filter over other discrete-time state estimators.  Since the state estimates are propagated with a small time step, accuracy of the prediction step is not degraded as the time interval between update steps is increased.  Errors due to disturbances or model inaccuracies still affect the accuracy of the estimate, so update steps must be performed sufficiently often to avoid large errors during transient events.
\begin{figure}[tb]
    \centering
    \includegraphics[width=3.1in]{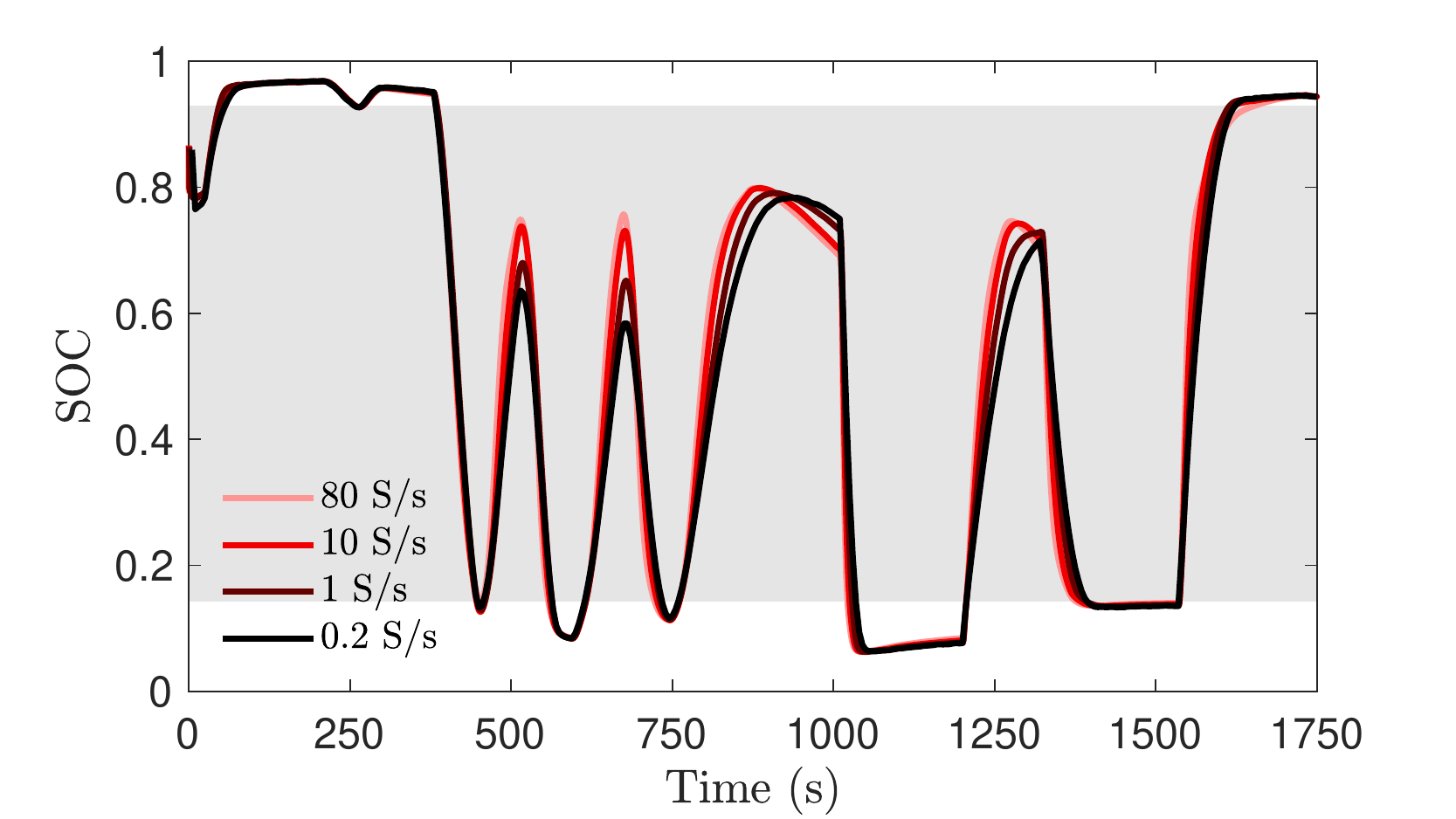}
    \caption{Estimated SOC for the three sample rates when TC3 is withheld.}
    \label{fig:SOC_exp_TC3_SR}
\end{figure}
\section{Conclusion} \label{sec:conclusion}

In this paper we designed and experimentally demonstrated a state-dependent Riccati equation (SDRE) filter for nonlinear state and state of charge estimation in a phase-change thermal energy storage device.  Unlike other state estimators for nonlinear systems, the SDRE filter uses a linear parameter-varying system model which offers several advantages.  First, the LPV model can be computed quickly using graph-based methods, and linearization is not required.  Second, the SDRE filter is generalizable to many TES architectures because the finite-volume heat transfer model can be applied to any heat conduction problem, even those with highly nonlinear dynamics and material properties.  Finally, the structure of the LPV model remains the same, so boundedness and detectability of the SDRE filter can be guaranteed in advance.

The SDRE filter is first tested in simulation to verify that the SOC estimate converges to the true SOC.  The SOC estimation error remains bounded within $\pm0.02$, even when the TES experiences rapid temperature changes or the PCM undergoes phase change.  In experimental tests, the SDRE filter is shown to accurately estimate temperatures at two locations in the TES.  The largest sources of error are model inaccuracies such as (i) the assumption that there is no heat transfer with the surroundings and (ii) a simplified phase-change model which does not account for hysteresis or undercooling during solidification.  Additionally, when the mass flow rate through the TES is zero, estimate accuracy is reduced because heat transfer along the length of the TES is typically dominated by mass transfer.  Despite the limitations of the finite volume model, the accuracy of the estimator is not significantly affected by the sample rate of the measurements; sample rates of 80 samples/s and 1 sample/s yield similar root-mean-square errors.  

Experimental validation of model-based state estimation techniques for fast-timescale thermal storage devices advances the continuing research on control strategies for hybrid transient thermal management systems. \revres{In future work, the boundedness and convergence of other Kalman-type and particle filters suitable for SOC estimation may be explored.} \revres{Additionally, investigation of the effects of model uncertainty and the previously discussed model inaccuracies has the potential to decrease estimation error.}

\revres{\section{Appendix} \label{sec:appendix}}
\revres{\begin{defn} [Connected Undirected Graph] A graph is a pair $\mathcal G = (\mathcal V,\mathcal E)$ where $\mathcal V$ is is a finite set of $n$ nodes or vertices  $\{v_1,\dots,v_n\}$,  and $\mathcal E \subset \{1,\dots,n\} \times \{1,\dots,n\}$ is a set of edges. The pair $(i,j)$ denotes the edge that links vertex $v_i$ to $v_j$. If $(i,j) \in \mathcal E \Leftrightarrow (j,i) \in \mathcal E$, the graph is said to be \emph{undirected}. An undirected graph is \emph{connected} if there exists a path (set of edges) between any two nodes. The undirected graph is \textit{connected} if there exists a path (set of edges) between any two nodes; the connected egde set is denoted with $\mathcal E_c$
\end{defn}}

\revres{\begin{defn}[Adjacency and Laplacian Matrices] A (weighted) graph is fully described by the adjacency matrix, $\Lambda \in \mathbb{R}^{n \times n}$ where
\begin{align}
\Lambda_{i,j} = 
\begin{cases}\lambda_{i,j} & \text{if } (i,j) \in \mathcal E \\
0& \text{otherwise};  \end{cases}
\end{align} 
where $\lambda_{i,j}$ is the \textit{weight} assigned to any of its edges $(i,j)$. Consider the diagonal matrix $D$ whose diagonal entry $D_{i,i}=\sum_{j=1}^{n}\lambda_{i,j}$. $D$ contains the  total egdes' weight that is incident to each vertex, and is termed the \textit{degree matrix}. The graph Laplacian is defined as $L = D-\Lambda$.
\end{defn}}

\revres{\begin{defn}[Integral Graph] \label{def:int_graph} Given a parameter-varying graph $\mathcal G(\rho)=$ $(\mathcal V, \mathcal E(\rho))$ for some time-dependent parameter $\rho$, the integral graph of $\mathcal G(\rho)$ on $[0, \infty)$ is a constant graph $\bar{\mathcal G}_{[0, \infty)}:=(\mathcal V, \bar{\mathcal E})$ where $\mathcal V$ is the same vertex set of $\mathcal G(\rho)$, and the adjacency matrix is defined by
$$
\bar{\Lambda}_{i,j}= \begin{cases}1, & \text { if } \int_0^{+\infty} \lambda_{i,j}(\rho) d t=\infty \\ 0, & \text { if } \int_0^{+\infty} \lambda_{i,j}(\rho) d t<\infty\end{cases}
$$
\end{defn}}

\revres{Consider the system where the dynamics of each node $x_i$ is given by
\begin{align}
    \dot{x}_i &=\sum_{j \in \mathcal{N}(i)} \lambda_{i,j}(\rho)\left(x_j - x_i\right) ~~ i\in\{1, \dots, n\}, \label{eq:sys_graph}
\end{align}
and $\lambda_{i,j}(\rho)$ is the weight on the vertex if $(i,j) \in \mathcal{E}$ (otherwise, $\lambda_{i,j}=0$). This system corresponds to a weighted graph and can be represented in state-space form as
\begin{align}
    \dot{x}&= -L(\rho)x, \label{eq:sys}
\end{align}
where $L$ is the Laplacian. Lemma \ref{lemma:sys_cons}  holds for the system \cite{cao_consensus_2011}.}

\revres{\begin{lemma}[\cite{cao_consensus_2011}]\label{lemma:sys_cons}
The dynamics of \eqref{eq:sys} implements consensus, that is, $\lim _{t \rightarrow+\infty} x(t) \in \operatorname{span}\{\mathbf{1}\}$ if and only if $\bar{\mathcal G}_{[0, \infty)}$ is connected.
\end{lemma}}

\revres{Lemma \ref{lemma:sys_cons2}  follows directly from Definition \ref{def:int_graph}.}

\revres{\begin{lemma}\label{lemma:sys_cons2} \eqref{eq:sys} achieves consensus if its corresponding parameter-varying graph $\mathcal G(\rho)$ is connected $\forall t$, that is, $\lambda_{i,j}(\rho) \geq \lambda_l >0$ for $(i,j)\in \mathcal E_c(\rho) ~\forall t$.
\end{lemma}}

\revres{Suppose the system in \eqref{eq:sys} has a piecewise constant output $y_{{k}}\in \mathbb{R}^{m}$ over time interval $\Delta t = t_{{k+1}} - t_{k}$, the discrete-time formulation of the system may be written as}
\revres{\begin{subequations}\label{eq:sys_disc}
\begin{align}
        x_{{{k}}+1} &= \Phi_{k}x_{{k}} \\
    y_{{k}} &= Cx_{{k}}
\end{align}
\end{subequations}}
\revres{where $x_{k} \in \mathbb{R}^{n}$ is the state vector, $\Phi_{k} \in \mathbb{R}^{n\times n}$  and $C \in \mathbb{R}^{m \times n}$are the system state and output matrices respectively. Further, for two positive successive integers $k, {l}$ (${l}\geq k$), let the \textit{state transition matrix} be defined as $\Phi_{{l}|k} = \Phi_{{l}|{l}-1}\Phi_{{l}-1|k}$ where $\Phi_{k|k} = I_n$ and $\Phi_{k+1|k} = \Phi_{k}$. Then the system is said to be uniformly detectable if Definition \ref{def:detectability} holds.} 

\revres{\begin{defn}[Uniform Detectability \cite{anderson_detectability_1981}] \label{def:detectability} The pair $(A_{k},C)$ is uniformly detectable if there exist integers $p$, $q \geq 0$, and some $0 \leq a<1$, $b>0$ such that whenever
\begin{align}
\left\|\Phi_{k+q|k}\zeta\right\| & \geq  a\left\|\zeta\right\|
\end{align}
{for some} $\zeta \in \mathbb{R}^{n}$ and $k$, then
\begin{align}
\zeta^TW_{k+r|k}\zeta \geq b\zeta^T\zeta,
\end{align}
where $W_{k+r|k}$  is the discrete-time observability gramian given by
\begin{align}\label{eq:gramian_disc}
   W_{k+r|k} \coloneqq \sum^{k+q}_{i = k}\Phi_{{i}|k}^{T}C^TC\Phi_{{i}|k}.
\end{align}
\end{defn}}

\revres{\begin{lemma}
$(\Phi_{k}, C)$ is uniformly detectable if $C$ has at least one non-zero row sum.
\end{lemma}}

\revres{\begin{proof}
Following Definition \ref{def:detectability}, we need to evaluate {for what $v$} {
\begin{align}
    \frac{\left\|\Phi_{k+q|k}v\right\|}{\left\|v\right\|} \geq a ,  ~ a \in [0,1)\label{eq:detect_norm}
\end{align}}
 and show that for such $v$,
 \begin{align}
v^TW_{k+r|k}v \geq bv^Tv ~~ (b>0)\label{eq:detect_2norm}
\end{align}
for {some integer} $q$ {$\geq 0$}.
Since $\mathcal G(\rho)$ implements consensus, that is, $\lim _{t \rightarrow+\infty} x(t) = x^*\in \operatorname{span}\{\mathbf{1}_n\}$; for any $t_{k}\geq0$ ($t_{k} = k\Delta t$ ) where $ x(t_{k}) \notin \operatorname{span}\{\mathbf{1}_n\}, ~\exists$ ${a \in [0,1)}$ and $t_{{p+k}}$ with $p>0$ such that
\begin{align} \label{eq:x_ineq}
    \left\|x(t_{{p+k}}) - x^*\right\| < {a}\left\|x(t_{k}) - x^*\right\| .
\end{align}
Without loss of generality, let $v = x(t_{k}) - x^*$. \eqref{eq:x_ineq} implies $\exists ~\Phi_{k+q|k} $ such that 
\begin{align}
    \frac{\left\|\Phi_{k+q|k}v\right\|}{\left\|v\right\|} < {a} , ~~\forall ~ v \notin \operatorname{span}\{\mathbf{1}_n\} 
\end{align}
Thus, we only need to show that \eqref{eq:detect_2norm} is satisfied with $v= \mathbf{1}_n$. Since $C$ has at least one non-zero row sum, $\exists$ {some}  $b {>0} $ such that
\begin{align}
v^TW_{k+r|k}v  &= v^T\sum^{k+q}_{i = k}\Phi_{{i}|k}^{T}C^TC\Phi_{{i}|k}v \nonumber \\&= {\begin{cases} 
\mathbf{1}_n^TC^TC\mathbf{1}_n
~ (q=0)\\q \mathbf{1}_n^TC^TC\mathbf{1}_n\end{cases}} \geq b>0.
\end{align}
\end{proof}}

\begin{acknowledgment}
This work is supported by the U.S. Office of Naval Research Thermal Science and Engineering Program under contract number N00014-21-1-2352.
\end{acknowledgment}


\bibliographystyle{asmems4}
\bibliography{references}

\end{document}